\documentclass[pdflatex,sn-nature]{sn-jnl}

\usepackage{amsmath,amssymb,mathtools}
\usepackage{graphicx}
\usepackage{booktabs}
\usepackage{tabularx}
\usepackage{algorithm}
\usepackage{algpseudocode}
\usepackage{braket}
\usepackage{bm}
\usepackage{placeins}
\usepackage{xcolor}
\usepackage{enumitem}
\usepackage{float}

\graphicspath{{fig/}}
\hypersetup{
  hidelinks,
  hypertexnames=false,
  bookmarksdepth=2,
  pdftitle={Streaming the partial-transpose moment hierarchy with order-independent quantum memory},
  pdfauthor={Junxiang Huang, Xiaoyang Wang, Xiao Yuan, and Yukun Zhang},
  pdfsubject={Partial transpose moment estimation under active memory constraints},
  pdfkeywords={partial transpose moments, copy complexity, qubit reuse, active quantum memory}
}

\theoremstyle{thmstyleone}
\newtheorem{theorem}{Theorem}
\newtheorem{proposition}{Proposition}
\newtheorem{lemma}{Lemma}
\theoremstyle{thmstyletwo}

\theoremstyle{thmstylethree}

\newcommand{\Tr}{\operatorname{Tr}}
\newcommand{\E}{\mathbb{E}}
\newcommand{\Prob}{\mathbb{P}}
\newcommand{\epsmom}{\varepsilon_{\mathrm{mom}}}
\newcommand{\ptperm}[1]{\Pi_{#1}^{(\mathrm{PT})}}
\newcommand{\SA}[1]{\mathcal{S}_A^{(#1)}}
\newcommand{\SB}[1]{\mathcal{S}_B^{(#1)}}
\newcommand{\Ntotal}{N_{\mathrm{total}}}
\newcommand{\Mshot}{M_{\mathrm{shot}}}

\DeclareMathOperator{\erfop}{erf}
\DeclareMathOperator{\erfc}{erfc}

\begin{document}

\title[Streaming partial-transpose moments]{Streaming the partial-transpose moment hierarchy with order-independent quantum memory}

\author[1,2]{\fnm{Junxiang} \sur{Huang}}
\author[3,4]{\fnm{Xiaoyang} \sur{Wang}}
\author*[2]{\fnm{Xiao} \sur{Yuan}}\email{xiaoyuan@pku.edu.cn}
\author*[2]{\fnm{Yukun} \sur{Zhang}}\email{yukunzhang@stu.pku.edu.cn}

\affil[1]{\orgdiv{School of Computer Science}, \orgname{Peking University}, \orgaddress{\city{Beijing}, \postcode{100871}, \country{China}}}
\affil[2]{\orgdiv{Center on Frontiers of Computing Studies}, \orgname{Peking University}, \orgaddress{\city{Beijing}, \postcode{100871}, \country{China}}}
\affil[3]{\orgdiv{RIKEN Center for Interdisciplinary Theoretical and Mathematical Sciences}, \orgname{RIKEN}, \orgaddress{\city{Wako}, \postcode{351-0198}, \country{Japan}}}
\affil[4]{\orgdiv{RIKEN Center for Computational Science}, \orgname{RIKEN}, \orgaddress{\city{Kobe}, \postcode{650-0047}, \country{Japan}}}

\abstract{Partial transposition reveals mixed-state entanglement through the spectrum of
\(\rho_{AB}^{T_B}\), but this operator is not generally a physical state. Its power
moments are accessible through collective permutations, whose coherent width grows
with moment order. Here we introduce a streaming protocol that processes fresh copies
sequentially while retaining one copy and a reusable ancilla. By realizing opposite
cyclic permutations on the two subsystems, each measurement record yields unbiased
estimators of all moments \(p_2,\ldots,p_K\) through successive outcome products.
For an \(n\)-qubit state, the protocol uses \(2n+1\) active qubits independent of
\(K\) and \(O(K\log K/\epsilon_{\rm mom}^2)\) copies, matching
\(\Omega(K/\epsilon_{\rm mom}^2)\) lower bounds up to a logarithmic factor,
including a PT-specific isospectral construction. Higher-order moments strengthen
thermal-state entanglement certification toward the exact PPT threshold. Simulations
cover pure and mixed states, and a cloud quantum processor implements a three-qubit
instance.}

\maketitle

\section{Introduction}
\label{sec:introduction}

Entanglement in mixed quantum systems cannot generally be inferred from the purity or spectrum of the density operator alone. Partial transposition provides a widely used alternative: negative eigenvalues of \(\rho_{AB}^{T_B}\) reveal negative partial transpose (NPT) entanglement, and their total weight determines the entanglement negativity~\cite{Peres1996,Horodecki1996,horodecki2001separability,VidalWerner2002,plenio2007introduction}. The central experimental difficulty is that partial transposition is not a physical quantum channel, so \(\rho_{AB}^{T_B}\) cannot in general be prepared as a state. Its power moments
\begin{equation}
p_j(\rho_{AB})=\Tr\!\left[(\rho_{AB}^{T_B})^j\right],\qquad j\geq 2,
\label{eq:intro_moments}
\end{equation}
provide experimentally accessible summaries of this otherwise inaccessible spectrum. These partial-transpose (PT) moments support entanglement certification, many-body phase diagnostics, two-qubit characterization and moment-based quantification~\cite{elben2020mixed,neven2021symmetry,Yu2021PTMoments,Carrasco2024PTPhases,vermersch2024manybody,Ali2023PTMinors,BradshawLaBorde2025PPTMoments,Zhang2022TwoQubitPTMoments,Zhang2023VisualPTMoments,TarabungaHaug2025PTRealignment}.

A PT moment can be written as the expectation value of a permutation on multiple copies~\cite{HorodeckiEkert2002DirectDetection,ekert2002direct,LeiferLindenWinter2004PolynomialInvariants,carteret2005noiseless}. In the conventional coherent implementation, measuring the \(j\)th moment requires all \(j\) copies to coexist in the quantum processor. Estimating the hierarchy through order \(K\) therefore produces a coherent-width requirement that grows with \(K\), even when fresh states can be prepared repeatedly. The relevant limitation is not necessarily the total number of states that can be prepared, but how many copies must coexist coherently.

Several approaches redistribute these resources. Direct collective permutation measurements provide coherent access to nonlinear multi-copy observables. Randomized measurements and classical shadows replace coherent multi-copy control with repeated single-copy or shallow measurements and classical post-processing, and can reuse classical data for several quantities under protocol-dependent assumptions~\cite{daley2012measuring,islam2015measuring,brydges2019probing,liu2022detecting,huang2020predicting,elben2023randomized,gray2018machine,zhou2020single}. Qubit-reuse protocols reduce the active width for ordinary moments, and recent work studies simultaneous estimation of ordinary moments or broader nonlinear functionals~\cite{corcoles2021exploiting,yirka2021qubit,decross2023qubit,shi2025near,chen2025simultaneous}. Multivariate trace and trace-polynomial methods offer further depth--memory trade-offs~\cite{Quek2024multivariatetrace,Yao2024qsf,RicoHuber2024TracePolynomials,Rico2025StateWitnessContraction}. An explicit circuit that estimates the partial-transpose moment hierarchy from shared measurement records with active quantum memory independent of \(K\) remains missing, because partial transposition changes the underlying permutation structure.

The difference from ordinary state moments is structural. Ordinary moments are generated by one cyclic permutation acting in the same direction on all degrees of freedom. Partial transposition reverses the contraction direction on subsystem \(B\) relative to subsystem \(A\),
\begin{equation}
\ptperm{j}=\SA{j}\otimes(\SB{j})^{-1}.
\label{eq:intro_ptperm}
\end{equation}
Moreover, \(\rho_{AB}^{T_B}\) cannot be prepared and supplied to an ordinary moment-estimation circuit. A sequential implementation must reproduce the two opposite subsystem cycles directly and preserve their cross-copy contractions as fresh copies enter the circuit.

We accomplish this with a streaming protocol: fresh copies are processed sequentially while one copy is retained in memory. One bipartite copy occupies a storage register, a fresh copy enters an incoming register at each layer, and a reusable ancilla coherently selects the required subsystem exchange. The ancilla is then measured and reset, and the incoming register is reloaded. Each execution returns outcomes \(x_1,\ldots,x_{K-1}\): \(x_1\) estimates \(p_2\), \(x_1x_2\) estimates \(p_3\), and \(x_1\cdots x_{j-1}\) estimates \(p_j\). Thus the data for lower moments remain available as the circuit proceeds to higher orders.

We establish an exact sequential realization of the PT permutation and simultaneous estimation of \(p_2,\ldots,p_K\) from shared measurement records. The circuit uses \(2n+1\) active qubits for an \(n\)-qubit bipartite state, independent of \(K\), and achieves total copy complexity \(O(K\log K/\epsmom^2)\) for uniform additive error \(\epsmom\). Universal and PT-specific lower bounds both scale as \(\Omega(K/\epsmom^2)\). An application to thermal mixed states shows that extending the shared hierarchy to higher orders progressively improves moment-based entanglement certification. Noiseless simulations test the protocol for both pure and mixed states, and an experiment on a cloud quantum processor realizes a three-qubit instance. More broadly, the result makes the partial-transpose moment hierarchy accessible through sequential processing: fresh copies are processed online, while a shared measurement record retains information about every moment order reached by the circuit.

\section{Results}
\label{sec:results}

\subsection{Multi-copy representation of partial-transpose moments}
\label{sec:sequential_access}

We assume access to fresh independent copies of an unknown bipartite state \(\rho_{AB}\). For an ordinary moment \(\Tr(\rho_{AB}^j)\), every degree of freedom follows the same cyclic permutation across the \(j\) copies. Partial transposition reverses this direction on subsystem \(B\), so the relevant multi-copy identity is
\begin{equation}
\Tr\!\left[(\rho_{AB}^{T_B})^j\right]
=\Tr\!\left[\ptperm{j}\rho_{AB}^{\otimes j}\right],
\qquad
\ptperm{j}=\SA{j}\otimes(\SB{j})^{-1}.
\label{eq:ptperm_identity}
\end{equation}
where \(\SA{j}\) follows the forward cycle on \(A\) and \((\SB{j})^{-1}\) follows the inverse cycle on \(B\). This subsystem-dependent reversal is the structural obstacle to applying an ordinary-moment reuse circuit directly. Figure~\ref{fig:topology} depicts the corresponding copy-index contractions; its line crossings are schematic connections between indices rather than physical wire intersections. A component derivation is given in Supplementary Note~1.

For \(j>2\), \(\ptperm{j}\) need not be Hermitian. Its expectation in Eq.~\eqref{eq:ptperm_identity} is nevertheless real because it equals a power moment of the Hermitian operator \(\rho_{AB}^{T_B}\). The sequential \(X\)-basis measurement below produces a real \(\{\pm1\}\)-valued estimator for this quantity.

\begin{figure}[t]
\centering
\includegraphics[width=\linewidth]{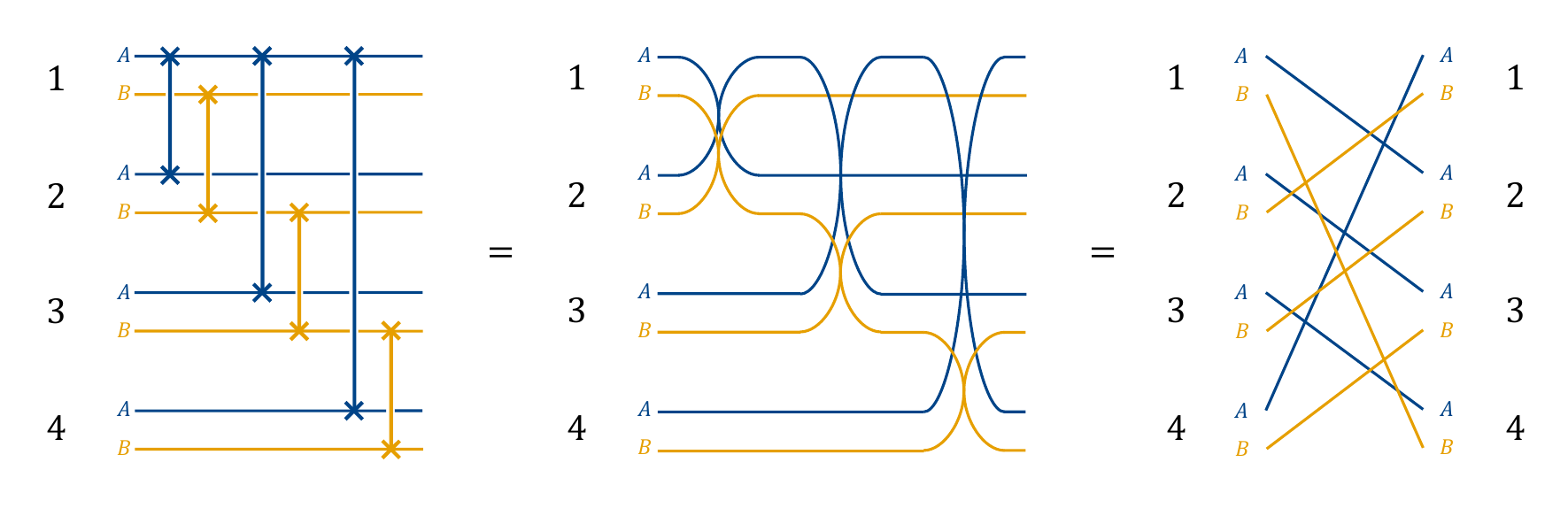}
\caption{\textbf{Multi-copy permutation for partial-transpose moments.} Subsystem \(A\) follows one cyclic direction, whereas subsystem \(B\) follows the inverse direction. This subsystem-dependent reversal distinguishes partial-transpose moments from ordinary state moments and is reproduced sequentially by the circuit in Fig.~\ref{fig:circuit}.}
\label{fig:topology}
\end{figure}

\subsection{Streaming circuit}
\label{sec:active_memory}

The circuit replaces the growing coherent width of the collective measurement by sequential processing of incoming copies. One bipartite copy is retained in an \(n\)-qubit storage register, while a fresh copy enters an \(n\)-qubit incoming-copy register at each layer. A reusable ancilla coherently selects the subsystem exchange required by the two opposite cycles. After each layer, the ancilla is measured and reset, and the incoming-copy register is reloaded. The retained storage register carries the contraction information forward to the next order.

The circuit execution produces a sequence of outcomes. Successive products of these outcomes estimate every moment order reached by the circuit, so extending an execution to a higher order preserves the lower-order estimates. Figure~\ref{fig:circuit} shows this sequential procedure. The exact ancilla-controlled operation and its storage-register recursion are given in Methods.

\begin{figure}[t]
\centering
\includegraphics[width=0.88\linewidth]{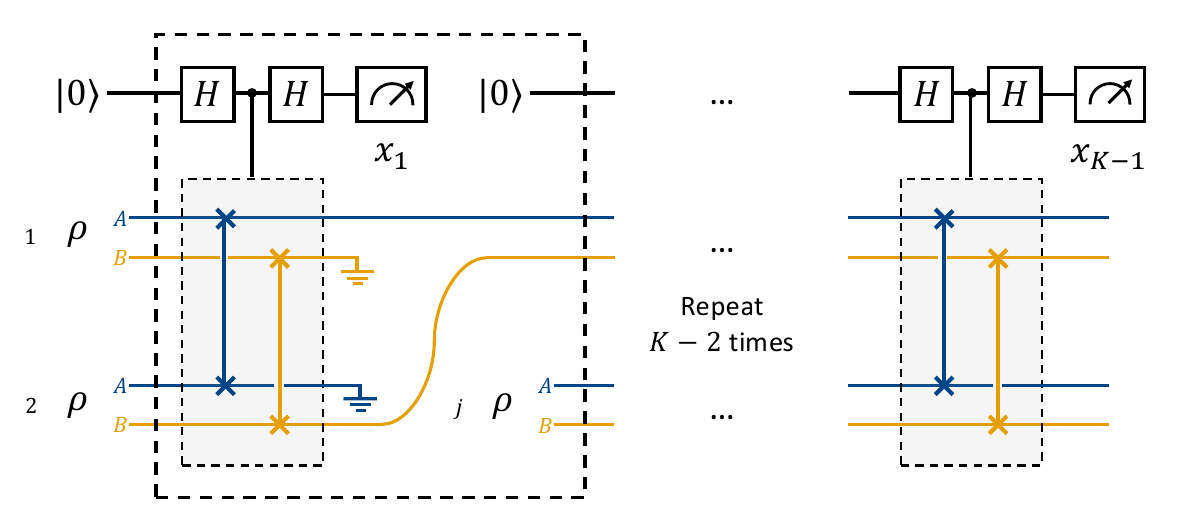}
\caption{\textbf{Streaming circuit for partial-transpose moments.} One bipartite copy is retained in the storage register while fresh copies enter the incoming register sequentially. At layer \(\ell\), a reusable ancilla controls the required subsystem exchanges, is measured to produce \(x_\ell\in\{\pm1\}\), and is then reset before the incoming register is reloaded. For one depth-\(K\) measurement record, the single-shot estimators are \(\widehat v_2=x_1\), \(\widehat v_3=x_1x_2\), and \(\widehat v_j=\prod_{\ell=1}^{j-1}x_\ell\). Averaging these prefix products over repeated records gives \(\widehat p_j\) for every \(j=2,\ldots,K\), while only \(2n+1\) qubits remain active.}
\label{fig:circuit}
\end{figure}

\subsection{One measurement record estimates all moment orders}
\label{sec:shared_records}

One depth-\(K\) execution returns ancilla outcomes \(x_1,\ldots,x_{K-1}\in\{\pm1\}\). The first outcome \(x_1\) estimates \(p_2\), the product \(x_1x_2\) estimates \(p_3\), and \(x_1x_2x_3\) estimates \(p_4\). In general, the product of the first \(j-1\) outcomes, or prefix parity, is
\begin{equation}
\widehat v_j=\prod_{\ell=1}^{j-1}x_\ell,\qquad j=2,\ldots,K.
\label{eq:prefix_estimator}
\end{equation}
The full hierarchy is extracted from shared trajectories, without collecting a separate order-specific data set for each \(j\). The estimators within one trajectory are statistically correlated, but each coordinate remains unbiased.

\begin{lemma}[Sequential parity estimator]
\label{lem:seq}
For each \(j=2,\ldots,K\), the cumulative parity in Eq.~\eqref{eq:prefix_estimator} satisfies
\begin{equation}
\E[\widehat v_j]
=\Tr\!\left[\ptperm{j}\rho_{AB}^{\otimes j}\right]
=\Tr\!\left[(\rho_{AB}^{T_B})^j\right].
\label{eq:unbiased}
\end{equation}
\end{lemma}

The proof identifies a closed recursion for the parity-weighted storage operator with the corresponding power of the partially transposed state. Tracing that identity gives Eq.~\eqref{eq:unbiased}. Methods state the recursion, and Supplementary Note~2 gives the complete matrix-element induction.

\subsection{Quantum memory and copy complexity}
\label{sec:bounds}

Exactly two \(n\)-qubit registers and one ancilla are occupied during the circuit. The peak active quantum memory is therefore
\begin{equation}
Q_{\mathrm{peak}}=2n+1,
\label{eq:peak_memory}
\end{equation}
independent of the maximum moment order \(K\). Averaging \(\Mshot\) shared trajectories and controlling all \(K-1\) output coordinates gives the main estimation guarantee.

\begin{theorem}[Simultaneous PT-moment estimation]
\label{thm:main_upper}
For every integer \(K\geq2\) and every \(\epsmom>0\), there exists a sequential protocol that, given copy access to an unknown \(n\)-qubit bipartite state \(\rho_{AB}\), never uses more than \(2n+1\) active qubits, performs
\(
\Mshot=O(\log K/\epsmom^2)
\)
independent depth-\(K\) circuit executions, and outputs estimates \(\{\widehat p_j\}_{j=2}^K\) such that
\[
\Pr\!\left[
\max_{2\leq j\leq K}
\left|\widehat p_j-p_j(\rho_{AB})\right|
\leq\epsmom
\right]
\geq\frac{2}{3}.
\]
Equivalently,
\begin{equation}
\Ntotal=K\Mshot=O\!\left(\frac{K\log K}{\epsmom^2}\right).
\label{eq:upper_complexity}
\end{equation}
\end{theorem}

The logarithm arises from uniform maximum-error control over the complete hierarchy. For \(\Mshot\) trajectories, Hoeffding's inequality~\cite{hoeffding1963probability} and a union bound give
\begin{equation}
\Prob\!\left[\max_{2\leq j\leq K}|\widehat p_j-p_j|>\epsmom\right]
\leq2(K-1)e^{-\Mshot\epsmom^2/2}.
\label{eq:uniform_concentration}
\end{equation}
Thus \(\Mshot\geq2\epsmom^{-2}\log(6K)\) suffices, and each depth-\(K\) trajectory consumes \(K\) copies.

\begin{proposition}[Universal minimax converse]
\label{prop:universal}
Fix \(K\geq2\). There is an absolute constant \(c>0\) such that, for every \(\epsmom\in(0,c)\), any uniformly valid estimator of \(p_K\) with additive error \(\epsmom\) and success probability at least \(2/3\) must use
\begin{equation}
\Ntotal=\Omega\!\left(\frac{K}{\epsmom^2}\right)
\end{equation}
copies.
\end{proposition}

\begin{proposition}[PT-specific converse]
\label{prop:ptspecific}
Fix \(K\geq3\). The same lower bound holds when the input is promised to lie in the two-qubit family
\begin{equation}
\ket{\psi_\theta}=\cos\theta\ket{00}+\sin\theta\ket{11},
\qquad
\theta\in\left[\frac{1}{5\sqrt K},\frac{1}{4\sqrt K}\right].
\label{eq:npt_family}
\end{equation}
All ordinary moments satisfy \(\Tr(\rho_\theta^m)=1\), whereas the higher PT moments vary with \(\theta\).
\end{proposition}

The universal converse uses a commuting separable family to establish a worst-case barrier for unrestricted estimators. The second construction is specific to partial transposition: the states in Eq.~\eqref{eq:npt_family} are globally isospectral and all their ordinary moments are identical, yet their partial-transpose moments vary. It therefore isolates information that cannot be recovered from ordinary moment estimation. Together, the bounds determine the leading linear dependence on \(K\) and the inverse-square dependence on the target accuracy, leaving a logarithmic gap for uniform estimation of the complete hierarchy. Supplementary Notes~3 and 4 give the binary-discrimination proofs.

Direct collective measurements trade increasing coherent width for direct access to nonlinear observables. Separate qubit-reuse circuits reduce the width but collect order-specific data, whereas the present protocol combines order-independent active memory with shared measurement records for all orders. Randomized and shadow-based approaches avoid coherent multi-copy control under different measurement and post-processing assumptions. Table~\ref{tab:resources} summarizes these resource choices.

\begin{table}[!htbp]
\caption{\textbf{Resource trade-offs for estimating partial-transpose moments through order \(K\).} Direct collective measurements increase coherent width with moment order. Separate reuse circuits reduce the width but require order-specific data. The present protocol combines \(K\)-independent active quantum memory with simultaneous estimates of all orders \(2,\ldots,K\) from shared records. Randomized and shadow-based protocols avoid coherent multi-copy control under different measurement and post-processing assumptions.}
\label{tab:resources}
\centering
\footnotesize
\setlength{\tabcolsep}{4pt}
\renewcommand{\arraystretch}{1.16}
\begin{tabularx}{\textwidth}{@{}>{\raggedright\arraybackslash}p{0.18\textwidth}>{\raggedright\arraybackslash}p{0.14\textwidth}>{\raggedright\arraybackslash}p{0.17\textwidth}>{\raggedright\arraybackslash}p{0.19\textwidth}>{\raggedright\arraybackslash}X@{}}
\toprule
Approach & Active quantum memory & Data reuse across orders & Copy or measurement scaling & Main requirement \\
\midrule
Direct multi-copy PT permutation~\cite{carteret2005noiseless} & \(Kn+1\) at order \(K\) & One order per circuit & \(O(K^2/\epsmom^2)\) for the hierarchy & Coherent control of multiple copies \\
Separate order-specific reuse circuits & \(2n+1\) & Separate data for each order & \(O(K^2/\epsmom^2)\) for the hierarchy & Repeated sequential circuits \\
Present streaming protocol & \(\mathbf{2n+1}\), independent of \(K\) & \(\mathbf{All}\) \(p_2,\ldots,p_K\) from shared records & \(\mathbf{O(K\log K/\epsmom^2)}\) & Sequential processing and dynamic operations \\
Randomized or shadow-based methods~\cite{elben2020mixed,huang2020predicting,elben2023randomized} & \(n\) system qubits & Protocol dependent; classical data may be reused & Protocol dependent & Single-copy or shallow measurements and classical processing \\
\botrule
\end{tabularx}
\par\smallskip\footnotesize
The \(O(K^2/\epsmom^2)\) entries sum order-specific copy costs at constant confidence for each order. A joint maximum-error guarantee for separate estimates also requires confidence allocation across the target orders. Copies, measurement settings and classical post-processing are distinct resources.
\end{table}

\FloatBarrier
\subsection{Entanglement information from the moment hierarchy}
\label{sec:diagnostics}

The shared hierarchy supports several analyses of the partially transposed spectrum. Few-moment criteria can test entanglement on structured state families, and sequences of moments can be compared across parameters to study phases and many-body entanglement~\cite{Yu2021PTMoments,MillerEisert2026,BradshawLaBorde2025PPTMoments,Carrasco2024PTPhases,vermersch2024manybody}. Two-qubit characterization and partial-transpose or realignment-moment diagnostics provide further examples~\cite{Zhang2022TwoQubitPTMoments,Zhang2023VisualPTMoments,TarabungaHaug2025PTRealignment}. Because every order is read from the same trajectories, several criteria can be evaluated without recollecting an order-specific data set.

To illustrate this reuse, set \(p_1=1\) and define the shifted moment matrices~\cite{Yu2021PTMoments}
\begin{equation}
B_r=\left[p_{a+b+1}\right]_{a,b=0}^{r}.
\label{eq:shifted_moment_matrix}
\end{equation}
The matrices \(B_1\), \(B_2\), and \(B_3\) use moments through \(p_3\), \(p_5\), and \(p_7\), respectively. If \(\rho_{AB}^{T_B}\) is positive semidefinite, then every \(B_r\) is positive semidefinite; consequently, \(\lambda_{\min}(B_r)<0\) certifies NPT entanglement.

We apply these criteria to the Gibbs state \(\rho_\beta=e^{-\beta H}/\Tr(e^{-\beta H})\) of the four-qubit open chain
\begin{equation}
H=\sum_{i=1}^{3}(X_iX_{i+1}+Z_iZ_{i+1})
\end{equation}
with a \(2|2\) bipartition, using moments through \(p_7\). Figure~\ref{fig:certification} shows that the exact PPT boundary is \(\beta_c\simeq0.394\), while the criteria using moments through \(p_7\), \(p_5\), and \(p_3\) cross zero at \(\beta_c\simeq0.436\), \(0.686\), and \(1.433\), respectively. Thus extending the same shared hierarchy progressively moves the certification boundary towards the exact threshold and detects entanglement at higher temperatures, corresponding to smaller \(\beta\). Panel \textbf{c} of Fig.~\ref{fig:certification} shows finite-sample convergence of the estimated minimum eigenvalues at \(\beta=1.2\), with all moments \(p_2,\ldots,p_7\) extracted from the same depth-\(7\) records.

\begin{figure}[htbp]
\centering
\includegraphics[width=\linewidth]{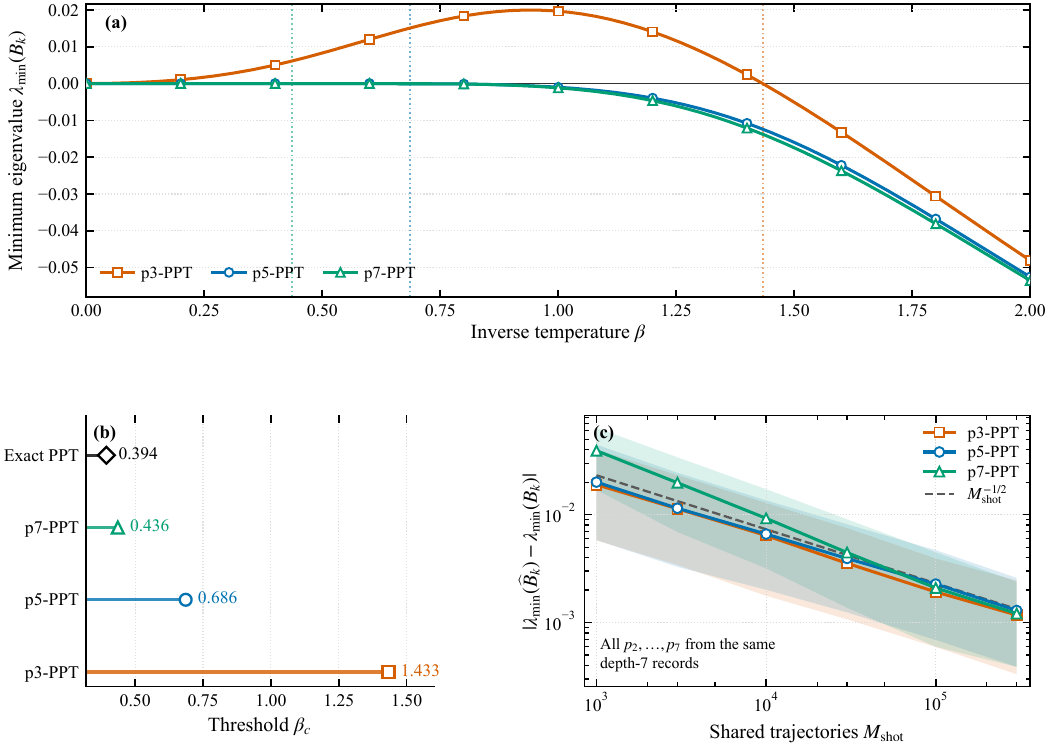}
\caption{\textbf{Entanglement certification from the shared partial-transpose moment hierarchy.} We consider the Gibbs state of a four-qubit open chain \(H=\sum_{i=1}^{3}(X_iX_{i+1}+Z_iZ_{i+1})\) with a \(2|2\) bipartition. \textbf{a}, Minimum eigenvalues of the shifted moment matrices \(B_1\), \(B_2\), and \(B_3\), corresponding respectively to criteria using moments through \(p_3\), \(p_5\), and \(p_7\). A negative minimum eigenvalue certifies NPT entanglement; coloured dotted lines mark the corresponding certification thresholds. \textbf{b}, Comparison with the exact PPT threshold. Increasing the maximum moment order moves the certification boundary towards the exact result. \textbf{c}, Finite-sample convergence at \(\beta=1.2\). All moments \(p_2,\ldots,p_7\) are extracted from the same depth-\(7\) measurement records. Lines show median absolute errors in the estimated minimum eigenvalues, shaded regions show the 16th--84th percentile ranges over \(2000\) independent data sets, and the dashed line indicates \(M_{\mathrm{shot}}^{-1/2}\) scaling.}
\label{fig:certification}
\end{figure}

If a diagnostic requires only a constant-size subset of moments, the circuit can be run to the largest required order and only those prefix products retained. If a broader spectral inference is desired, the complete measured hierarchy remains available for classical processing. For example, the negativity is defined as~\cite{VidalWerner2002}
\begin{equation}
\mathcal N(\rho_{AB})=\frac{\|\rho_{AB}^{T_B}\|_1-1}{2}
\end{equation}
Reconstructing the trace norm from finitely many moments introduces an additional classical approximation problem~\cite{rivlin2020chebyshev,childs2012hamiltonian}. Supplementary Note~5 analyzes polynomial and coherent linear-combination post-processing. The same measurement data can therefore support several moment-based analyses without repeating the quantum experiment for each order.

\FloatBarrier
\subsection{Numerical and hardware results}
\label{sec:validation}

We verified the simultaneous estimator across multiple orders in noiseless simulations. Figure~\ref{fig:numerical} compares exact PT moments with estimates obtained from shared trajectories for a six-qubit pure ground state and a three-qubit thermal mixed state. In both cases, all displayed orders come from the same trajectory ensemble, and the estimates agree with the exact values within sampling fluctuations.

\begin{figure}[htbp]
\centering
\includegraphics[width=\linewidth]{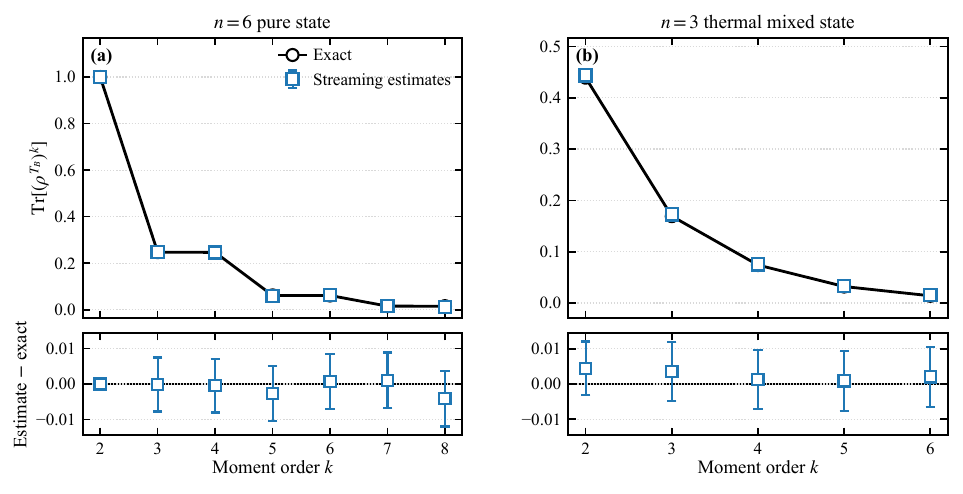}
\caption{\textbf{Numerical verification of simultaneous partial-transpose moment estimation.} Exact PT moments (black circles) and streaming estimates (blue squares) are shown for \textbf{a}, a six-qubit pure ground state up to \(K=8\), and \textbf{b}, a three-qubit mixed thermal state at inverse temperature \(\beta=1.2\) up to \(K=6\). All orders in each example are obtained from shared trajectories. Error bars show marginal \(95\%\) normal-approximation sampling intervals computed from the trajectory-level prefix estimators. The lower panels show estimated minus exact values; the residuals fluctuate around zero within sampling uncertainty.}
\label{fig:numerical}
\end{figure}

We also implemented the complete protocol for a three-qubit state through order \(K=5\) on the dynamic-circuit processor \texttt{ibm\_pittsburgh}. The implementation includes mid-circuit measurement, reset, reuse of the storage and incoming-copy registers, and error mitigation. Figure~\ref{fig:cloud} compares exact values, raw estimates and estimates obtained after variable-noise affine regression. The target data used \(60000\) shots across two folded circuits; readout calibration and regression training used the additional shots specified in Methods. For this instance, mitigation brings the estimates substantially closer to the exact values.

\begin{figure}[htbp]
\centering
\includegraphics[height=0.48\textheight,keepaspectratio]{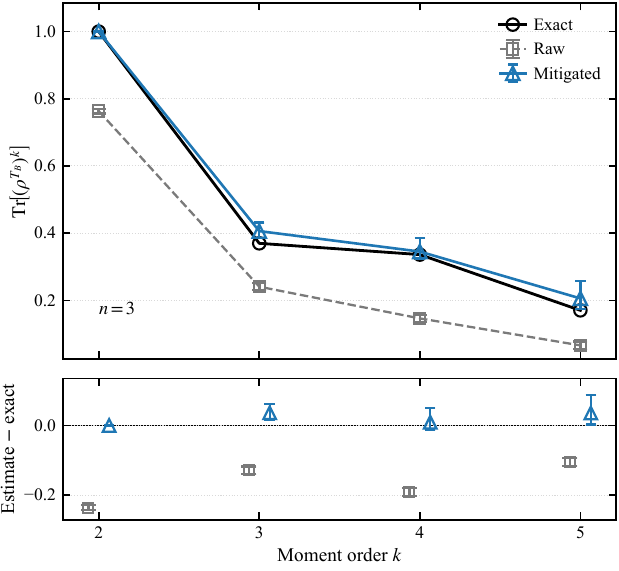}
\caption{\textbf{Dynamic-circuit hardware implementation.} Exact partial-transpose moments (black circles), raw estimates (grey squares), and mitigated estimates obtained with CDR-inspired variable-noise affine regression (blue triangles) are shown for a three-qubit instance through \(K=5\) on \texttt{ibm\_pittsburgh}. Within each folded circuit, all moment orders are extracted from the same streaming measurement records. The upper panel shows the estimated moments, and the lower panel shows the signed residuals, estimate minus exact. Error bars denote \(95\%\) percentile intervals obtained by multinomially bootstrapping all saved target, readout-calibration and regression-training count data. The two folded target circuits used \(60{,}000\) shots in total.}
\label{fig:cloud}
\end{figure}

\FloatBarrier

\section{Discussion}
\label{sec:discussion}
\vspace{6pt}

Partial transposition cannot generally be applied to prepare a physical input state, yet its opposite subsystem cycles can be evaluated sequentially. The streaming circuit turns the collective PT-permutation measurement into a process with one stored copy, one incoming copy and one reusable ancilla. The products within a single measurement record remain unbiased for every order reached by the circuit. This mechanism differs structurally from ordinary-moment reuse, which follows the same cyclic direction on all subsystems.

The protocol exchanges coherent width for sequential processing. Conventional collective measurements require an increasing number of simultaneously occupied qubits as the moment order grows; the protocol instead repeatedly prepares fresh states and retains two system registers. This trade-off is useful when state preparation is available but coherent quantum memory is scarce. Randomized and shadow-based methods offer a complementary route by avoiding coherent multi-copy control and reusing classical data under protocol-specific conditions. These approaches make different trade-offs between coherent control, measurement number and classical processing.

A shared hierarchy can support several moment-based diagnostics from the same quantum data. The four-qubit thermal-state example makes this reuse concrete. Criteria through \(p_5\) and \(p_7\) approach the exact PPT boundary substantially more closely than the \(p_3\) criterion, using the same depth-\(7\) records rather than separate order-specific data. The isospectral lower-bound family further shows that partial-transpose moments contain information absent from the ordinary spectrum: all ordinary moments remain fixed while the PT moments vary. The same sequential principle may motivate access to other multi-copy observables with subsystem-dependent permutations.

The \(2n+1\) count covers the active qubits of the moment-measurement protocol; preparing a mixed input through purification may require additional ancillas. The total number of copies and the sequential depth grow with \(K\), and reconstructing negativity introduces a separate approximation-stability cost. Closing the remaining logarithmic gap and improving robustness to depth, reset and mid-circuit errors are natural directions for future work. By separating quantum memory from moment order, the protocol provides a direct route to higher-order partial-transpose information on devices with limited coherent width.

\section{Methods}
\label{sec:methods}

The methods below formalize the estimation task, the sequential recursion and the statistical analysis, followed by the numerical and hardware procedures.

\subsection{Problem setting and resource definitions}

We work in the copy-access model, in which an algorithm may request fresh independent copies of an unknown bipartite state \(\rho_{AB}\) on an \(n\)-qubit Hilbert space \(\mathcal H_A\otimes\mathcal H_B\). The target hierarchy is
\begin{equation}
\mathbf p^{(K)}(\rho_{AB})=\bigl(p_2(\rho_{AB}),\ldots,p_K(\rho_{AB})\bigr).
\end{equation}
An estimator succeeds when
\begin{equation}
\Prob\!\left[\max_{2\leq j\leq K}|\widehat p_j-p_j(\rho_{AB})|\leq\epsmom\right]\geq\frac23.
\label{eq:success}
\end{equation}
The total number of copies counts every consumed input copy. Peak active quantum memory counts all simultaneously occupied protocol qubits, including the measurement ancilla, but excludes classical storage and classical post-processing. Adaptive circuits with mid-circuit measurement, reset and classical feed-forward are allowed. State-preparation ancillas are not included in the \(2n+1\) protocol count.

\subsection{Partial transpose permutation representation}

Write
\begin{equation}
\rho_{AB}=\sum_{i,j,\mu,\nu}\rho_{i\mu,j\nu}\ket i\!\bra j_A\otimes\ket\mu\!\bra\nu_B.
\end{equation}
After partial transposition on \(B\), multiplication of \(j\) copies contracts the \(A\) indices along a forward cycle and the \(B\) indices along the opposite cycle. The same component contraction is generated by \(\SA{j}\otimes(\SB{j})^{-1}\), proving Eq.~\eqref{eq:ptperm_identity}. Supplementary Note~1 records the term-by-term derivation and fixes the cyclic-shift convention.

\subsection{Sequential circuit and storage recursion}

Let \(S=S_A\otimes S_B\) denote the storage register and \(F=F_A\otimes F_B\) the incoming-copy register. Define
\begin{equation}
W_A=\mathrm{SWAP}_{S_A,F_A}\otimes I_{S_BF_B},\qquad
W_B=I_{S_AF_A}\otimes\mathrm{SWAP}_{S_B,F_B}.
\end{equation}
One recursive layer has the effective pre-measurement unitary
\begin{equation}
U_{\mathrm{layer}}=\ket0\!\bra0\otimes W_B+\ket1\!\bra1\otimes W_A.
\end{equation}
An \(X\)-basis ancilla measurement induces Kraus operators \(M_+=(W_A+W_B)/2\) and \(M_-=(W_B-W_A)/2\). The ancilla is then re-prepared and \(F\) is reloaded with a fresh copy. Repeating the layer produces the outcomes in Eq.~\eqref{eq:prefix_estimator}.

For an operator \(X\) on the storage register, let
\begin{equation}
\mathcal T_x(X)=\Tr_F\!\left[M_x(X\otimes\rho)M_x^\dagger\right],
\qquad x\in\{+1,-1\}.
\end{equation}
The parity-weighted storage operator starts from \(\Omega_0=\rho\) and satisfies
\begin{equation}
\begin{split}
\Omega_{\ell+1}
&=\sum_{x=\pm1}x\,\mathcal T_x(\Omega_\ell)\\
&=\frac12\Tr_F\!\left[
W_A(\Omega_\ell\otimes\rho)W_B+
W_B(\Omega_\ell\otimes\rho)W_A
\right].
\end{split}
\label{eq:omega_main}
\end{equation}
The target family is \(\Theta_\ell=((\rho^{T_B})^{\ell+1})^{T_B}\). Direct matrix-element evaluation shows that each cross term in Eq.~\eqref{eq:omega_main} maps \(\Theta_\ell\) to \(\Theta_{\ell+1}\). Since \(\Omega_0=\Theta_0=\rho\), induction gives \(\Omega_\ell=\Theta_\ell\), and taking the trace proves Lemma~\ref{lem:seq}. Supplementary Note~2 gives the complete algorithm and algebraic proof.

\subsection{Estimators and concentration}

For each trajectory \(s\), all estimators \(\widehat v_j^{(s)}\) are computed as prefix products of the same bitstring. Their expectations follow from Lemma~\ref{lem:seq}. Averaging over \(\Mshot\) independent trajectories and applying Hoeffding's inequality gives the single-coordinate tail bound. Applying a union bound to \(K-1\) coordinates yields Eq.~\eqref{eq:uniform_concentration}. Each trajectory consumes exactly \(K\) copies, which gives Eq.~\eqref{eq:upper_complexity}.

\subsection{Converse constructions}

The universal converse considers the commuting separable states
\begin{equation}
\rho_p=p\ket{00}\!\bra{00}+(1-p)\ket{01}\!\bra{01},
\end{equation}
for which partial transpose is trivial. Choosing two nearby values near \(p=1-1/(2K)\) creates a PT-moment separation of at least \(4\epsmom\), while their Bernoulli relative entropy is \(O(\epsmom^2/K)\). Pinsker's inequality and the Holevo--Helstrom bound then imply Proposition~\ref{prop:universal}.

For Proposition~\ref{prop:ptspecific}, the PT spectrum of Eq.~\eqref{eq:npt_family} is
\begin{equation}
\left\{\cos^2\theta,\sin^2\theta,\sin\theta\cos\theta,-\sin\theta\cos\theta\right\}.
\end{equation}
On the stated \(K^{-1/2}\) window, the derivative of the \(K\)-th PT moment has magnitude at least \(\sqrt K/8\). Two parameters separated by \(32\epsmom/\sqrt K\) therefore have a \(4\epsmom\) moment gap, whereas their overlap is \(\cos(32\epsmom/\sqrt K)\). Pure-state discrimination gives \(\Ntotal\geq K/(9216\epsmom^2)\). Supplementary Notes~3 and 4 contain the complete constants and testing reductions.

\subsection{Entanglement-certification simulation}

For the application in Fig.~\ref{fig:certification}, we use the Gibbs state
\begin{equation}
\rho_\beta=\frac{e^{-\beta H}}{\Tr(e^{-\beta H})},\qquad
H=\sum_{i=1}^{3}(X_iX_{i+1}+Z_iZ_{i+1}),
\end{equation}
of a four-qubit open chain with a \(2|2\) bipartition.  The partial transpose is taken on \(B\) in the product basis \(\ket a_A\otimes\ket b_B\).  We set \(p_1=1\), compute the PT moments through \(p_7\), and form \(B_r=[p_{a+b+1}]_{a,b=0}^{r}\) for \(r=1,2,3\), corresponding respectively to criteria using moments through \(p_3,p_5,p_7\).

The exact curves are evaluated on \(\beta=0,0.01,\ldots,2.00\).  The Hamiltonian is diagonalized to construct \(\rho_\beta\), and Hermitian eigensolvers are used for \(\rho_\beta^{T_B}\) and the symmetrized moment matrices.  For the exact PPT boundary and each moment criterion, the first positive-to-negative crossing after \(\beta=0.05\) is bracketed on this grid, using \(10^{-12}\) as the numerical-zero tolerance, and refined with Brent's method.

The finite-sample calculation in Fig.~\ref{fig:certification}\textbf{c} uses \(\beta=1.2\), \(K=7\), and
\[
M_{\mathrm{shot}}\in\{1000,3000,10000,30000,100000,300000\}.
\]
For each value, \(2000\) independent data sets are drawn by multinomial sampling from the exact joint distribution of complete depth-\(7\) records.  Thus neither the individual ancilla outcomes nor the different moment orders are sampled independently: within every data set, \(\widehat p_2,\ldots,\widehat p_7\) are prefix products of the same sampled records.  The matrices \(\widehat B_r\) are formed from this shared estimated hierarchy with \(p_1=1\), and the plotted quantity is
\[
\left|\lambda_{\min}(\widehat B_r)-\lambda_{\min}(B_r)\right|.
\]
The lines show medians and the bands show the \(16\)th--\(84\)th percentile range.  The \(M_{\mathrm{shot}}^{-1/2}\) reference is normalized at \(M_{\mathrm{shot}}=10000\) to the geometric mean of the three median absolute errors at that point.  The base random seed is \(20260726\).  The complete calculation is reproduced by the entanglement-certification program supplied as Supplementary Software.

\subsection{Numerical validation}

The noiseless checks use the open one-dimensional Hamiltonian
\begin{equation}
H=\sum_{i=1}^{n-1}(X_iX_{i+1}+Z_iZ_{i+1}).
\end{equation}
For \(n=6\), we use its ground state with a \(3|3\) bipartition and estimate through \(K=8\). For \(n=3\), we use the Gibbs state \(e^{-1.2H}/\Tr(e^{-1.2H})\), prepare it through purification, adopt a \(1|2\) bipartition and estimate through \(K=6\). The trajectory count is \(\lceil30000\ln K\rceil\), giving \(62384\) and \(53753\) shots, respectively. Exact values are computed from the corresponding state vectors or density matrices. The simulator seeds are \(15\) and \(14\) for the pure and mixed cases, respectively, and the transpiler seed is \(1234\).

For moment order \(j\), the trajectory-level prefix estimator \(v_j\) takes values in \(\{-1,+1\}\).  The displayed error bar is the marginal normal-approximation interval
\begin{equation}
\widehat p_j\ \pm\ 1.96
\sqrt{\frac{\widehat{\operatorname{Var}}(v_j)}{M_{\mathrm{shot}}}},
\qquad
\widehat{\operatorname{Var}}(v_j)
=\frac{M_{\mathrm{shot}}(1-\widehat p_j^2)}{M_{\mathrm{shot}}-1}.
\end{equation}
All orders in a panel are extracted from one common ensemble, so their sampling correlations are retained in the generated records, although the plotted intervals are coordinate-wise marginal intervals.

\subsection{Cloud implementation and error mitigation}

The implementation uses seven physical qubits on \texttt{ibm\_pittsburgh}, with physical indices \((63,64,65,66,67,77,85)\). A fixed three-qubit hardware-efficient ansatz, with rotation angles generated from pseudorandom seed \(42\), prepares the target state. Moments through \(K=5\) are obtained from two folded target circuits with noise factors \(1\) and \(3\), each sampled with \(30000\) shots. Two ancilla-readout calibration circuits also use \(30000\) shots each.

Readout assignment correction and Pauli gate twirling are applied, followed by a variable-noise affine regression model inspired by Clifford data regression~\cite{Czarnik2021cdr}. The regression set contains \(70\) classically simulable three-qubit states generated with seeds \(1000,\ldots,1069\). Each state is run at the two folding factors, producing \(140\) training circuits with \(6000\) shots each. For each moment order, an affine ridge model uses the two folded estimates after feature standardization, with ridge parameter \(10^{-2}\). The target circuit is transpiled once with Qiskit's optimization level \(3\) and its layout is reused~\cite{Qiskit}; the resulting instruction-set circuit contains \(186\) CZ gates before folding. Supplementary Note~7 records the preparation angles, calibration table and device map.

Uncertainty bars are obtained from \(10{,}000\) bootstrap replicates of all saved count histograms.  The base random seed is \(20260726\).  In every replicate, the two calibration histograms, all \(140\) regression-training histograms and the two target histograms at folding factors \(1\) and \(3\) are independently multinomial-resampled at their original shot counts.  The readout assignment matrix and its outcome weights are rebuilt from the resampled calibration data; all training features are recomputed and standardized; the affine ridge model for every moment order is refitted; and the two resampled target features are used to obtain raw and mitigated estimates.  Each target histogram is sampled once per replicate, so all moment orders retain their common-record correlations, and the raw and mitigated estimates use the same replicate-specific resamples and readout correction.

The saved histograms already aggregate the primitive-level Pauli-twirling realizations; the bootstrap does not separately resample the finite twirling ensemble.  The plotted limits are the marginal \(2.5\)th and \(97.5\)th percentiles, evaluated with linear quantile interpolation.  No physicality projection is applied.  The intervals capture saved-count fluctuations conditional on the executed circuit family and do not include device drift, training-family selection, transpilation uncertainty or finite twirling-ensemble selection.

\backmatter

\section*{Data availability}
The numerical data supporting the entanglement-certification analysis and noiseless validation, together with the original cloud hardware measurement counts, readout-calibration data, regression-training data, final figure-level data and the device calibration summary reported in Supplementary Table S1, constitute the Supplementary Data. The Supplementary Data described here are not included in this arXiv submission and are available from the corresponding authors upon reasonable request. Detailed device-calibration records are available from the corresponding authors upon reasonable request.

\section*{Code availability}
The code used to generate the entanglement-certification analysis and numerical validation, process the cloud hardware data, construct the bootstrap confidence intervals and produce the figures constitutes the Supplementary Software. The Supplementary Software described here is not included in this arXiv submission and is available from the corresponding authors upon reasonable request.

\section*{Acknowledgements}
We thank You Zhou for helpful discussions. This work is supported by the Quantum Science and Technology--National Science and Technology Major Project (Grant No. 2023ZD0300200), the Beijing Natural Science Foundation Z250004, NSAF (Grant No. U2330201), the National Natural Science Foundation of China Grant (No. 12361161602), and the Beijing Science and Technology Planning Project (Grant No. Z25110100810000). The numerical experiments of this work are supported by the High-performance Computing Platform of Peking University.

\section*{Author contributions}
J.H., Y.Z. and X.Y. conceived the project. J.H. led the theoretical development, numerical studies and manuscript preparation. X.W. carried out the hardware experiments. Y.Z. contributed to the theoretical analysis. Y.Z. and X.Y. supervised the project. All authors discussed the results and reviewed the manuscript.

\section*{Competing interests}
The authors declare no competing interests.

\clearpage
\section*{Supplementary Information}
\noindent This Supplementary Information contains the complete protocol, proofs, downstream reconstruction analysis, numerical specifications and cloud implementation details supporting the Article.

\noindent This Supplementary Information accompanies the Article ``Streaming the partial-transpose moment hierarchy with order-independent quantum memory.''

\setcounter{theorem}{0}
\setcounter{proposition}{0}
\setcounter{lemma}{0}
\setcounter{remark}{0}
\setcounter{definition}{0}
\setcounter{algorithm}{0}
\setcounter{section}{0}
\setcounter{equation}{0}
\setcounter{figure}{0}
\setcounter{table}{0}
\renewcommand{\thesection}{Supplementary Note \arabic{section}}
\renewcommand{\theequation}{S\arabic{equation}}
\renewcommand{\thefigure}{S\arabic{figure}}
\renewcommand{\thetable}{S\arabic{table}}

\section{Derivation of the PT moment identity}
\label{app:pt_identity}

In this supplementary note, we derive
\begin{equation}
\Tr\!\bigl[(\rho_{AB}^{T_B})^k\bigr]
=
\Tr\!\left[\bigl(\SA{k}\otimes(\SB{k})^{-1}\bigr)\rho_{AB}^{\otimes k}\right].
\label{eq:ptperm_identity_supp}
\end{equation}
To compare the two sides of \eqref{eq:ptperm_identity_supp}, it is enough to expand both in a product basis and inspect how the indices are contracted.
This keeps the derivation short and makes the action of the PT permutation explicit.

We write
\begin{equation}
\rho_{AB}=\sum_{i,j}\sum_{\mu,\nu}\rho_{i\mu,j\nu}\,
\ket{i}\!\bra{j}_A\otimes\ket{\mu}\!\bra{\nu}_B.
\end{equation}
Then the partial transpose with respect to subsystem $B$ is
\begin{equation}
\rho_{AB}^{T_B}
=
\sum_{i,j}\sum_{\mu,\nu}\rho_{i\nu,j\mu}\,
\ket{i}\!\bra{j}_A\otimes\ket{\mu}\!\bra{\nu}_B.
\end{equation}
Consequently, this gives
\begin{equation}
\label{eq:trace_expansion_supp}
\begin{split}
\Tr\!\bigl[(\rho_{AB}^{T_B})^k\bigr]
=
\sum_{\substack{i_1,\dots,i_k\\ \mu_1,\dots,\mu_k}}
\rho_{i_1\mu_2,i_2\mu_1}
\rho_{i_2\mu_3,i_3\mu_2}
\cdots
\rho_{i_k\mu_1,i_1\mu_k}.
\end{split}
\end{equation}
The $A$ indices are contracted along the forward cycle
\(
i_1\to i_2\to \cdots \to i_k \to i_1,
\)
while the $B$ indices are contracted along the inverse cycle
\(
\mu_1\leftarrow \mu_2\leftarrow \cdots \leftarrow \mu_k \leftarrow \mu_1.
\)
These are exactly the forward contractions on subsystem $A$ and the inverse contractions on subsystem $B$ implemented by the PT permutation.

Evaluating the multi-copy observable directly gives
\begin{equation}
\begin{split}
\Tr\!\left[\bigl(\SA{k}\otimes(\SB{k})^{-1}\bigr)\rho_{AB}^{\otimes k}\right]
=
\sum_{\substack{i_1,\dots,i_k\\ \mu_1,\dots,\mu_k}}
\rho_{i_1\mu_2,i_2\mu_1}
\rho_{i_2\mu_3,i_3\mu_2}
\cdots
\rho_{i_k\mu_1,i_1\mu_k},
\end{split}
\end{equation}
The two expressions coincide term by term, so the PT-permutation expectation reproduces the PT moment.
This is exactly the PT-moment identity in \eqref{eq:ptperm_identity_supp}.

\section{Complete protocol and proof of the sequential parity estimator}
\label{app:full_lemma}

For convenience, we first record the complete protocol in pseudocode form.

\begin{algorithm}[H]
\caption{Simultaneous Estimation of Partial Transpose Moments}
\label{alg:pt_estimation}
\begin{algorithmic}[1]
\Require Copy access to $\rho_{AB}$, target moment accuracy $\epsmom$, and maximal moment order $K$.
\Ensure Estimates $\{\widehat{p}_j\}_{j=2}^K$.
\State Choose the number of circuit executions $\Mshot = O(\log K/\epsmom^2)$.
\State Initialize accumulators $S_2,\dots,S_K\leftarrow 0$.
\For{$s=1$ to $\Mshot$}
    \State Execute the depth-$K$ sequential circuit.
    \State Within this execution, re-prepare the ancilla in $\ket{+}$ at each recursive layer and reload the incoming-copy register with a fresh copy of $\rho_{AB}$.
    \State Record ancilla outcomes $x_1^{(s)},\dots,x_{K-1}^{(s)}\in\{\pm1\}$.
    \State $v\leftarrow 1$.
    \For{$\ell=1$ to $K-1$}
        \State $v\leftarrow v\cdot x_\ell^{(s)}$.
        \State $S_{\ell+1}\leftarrow S_{\ell+1}+v$.
    \EndFor
\EndFor
\For{$j=2$ to $K$}
    \State $\widehat{p}_j \leftarrow S_j/\Mshot$.
\EndFor
\State \Return $\{\widehat{p}_j\}_{j=2}^K$.
\end{algorithmic}
\end{algorithm}

\begin{proof}[Proof of the sequential parity estimator]
We follow the same three-step structure as in the main text.
The proof first derives the one-layer storage recursion, then identifies the corresponding PT-permutation family, and finally takes the trace to recover the estimator expectation.
Let the one-copy Hilbert space be
\begin{equation*}
\mathcal{H}=\mathcal{H}_A\otimes\mathcal{H}_B,
\qquad
\rho:=\rho_{AB}.
\end{equation*}
At one recursive layer, the storage register is
\(
S=S_A\otimes S_B
\)
and the incoming-copy register is
\(
F=F_A\otimes F_B.
\)
Define the subsystem-swap operators on $S\otimes F$ by
\begin{equation}
W_A:=\mathrm{SWAP}_{S_A,F_A}\otimes I_{S_BF_B},
\qquad
W_B:=I_{S_AF_A}\otimes \mathrm{SWAP}_{S_B,F_B}.
\label{eq:swap_defs_supp}
\end{equation}
They satisfy
\begin{equation}
W_A^2=W_B^2=I_{SF},
\qquad
[W_A,W_B]=0.
\label{eq:swap_relations_supp}
\end{equation}
These two swaps encode the complementary routing of subsystems $A$ and $B$ through one recursive layer.

Let $Q$ denote the ancilla in the computational basis $\{\ket0,\ket1\}$.
The exact pre-measurement ancilla-controlled unitary of one recursive layer is
\begin{equation}
U_{\mathrm{layer}}
=
\ket0\!\bra0_Q\otimes W_B
+
\ket1\!\bra1_Q\otimes W_A.
\label{eq:layer_unitary_supp}
\end{equation}
This is the effective unitary implemented by the recursive circuit.
The ancilla-$\ket0$ branch contributes only $W_B$.
The ancilla-$\ket1$ branch contributes the full-copy swap $W_AW_B$ followed by the subsystem-$B$ swap, leaving $(W_AW_B)W_B=W_A$.

Let $\ket{x}_X$ denote the Pauli-$X$ eigenstate with eigenvalue $x\in\{+1,-1\}$, so that $\ket{+1}_X=\ket+$ and $\ket{-1}_X=\ket-$.
The ancilla is prepared in $\ket+_X$ and measured in the $X$ basis.
Define the ancilla-conditioned Kraus operators
\begin{equation}
M_x:={}_X\!\bra{x}\,U_{\mathrm{layer}}\,\ket+_X,
\qquad
x\in\{+1,-1\}.
\label{eq:kraus_def_supp}
\end{equation}
Equation~\eqref{eq:layer_unitary_supp} yields
\begin{equation}
M_{+1}=M_+=\frac{W_A+W_B}{2},
\qquad
M_{-1}=M_-=\frac{W_B-W_A}{2}.
\label{eq:kraus_pm_supp}
\end{equation}
Moreover,
\begin{equation}
\begin{split}
M_+^\dagger M_+ + M_-^\dagger M_-
&=
\frac{(W_A+W_B)^2+(W_B-W_A)^2}{4}\\
&=
\frac{2W_A^2+2W_B^2}{4}
=I_{SF}.
\end{split}
\end{equation}

We first package one measurement layer into the conditional branch maps that act on the storage register alone.
For each $x\in\{+1,-1\}$ and every operator $X$ on the storage register, we define
\begin{equation}
\mathcal{T}_x(X):=
\Tr_F\!\left[M_x\,(X\otimes \rho_F)\,M_x^\dagger\right],
\label{eq:branch_map_supp}
\end{equation}
where $\rho_F=\rho$ is the fresh copy loaded at that layer.
These maps describe the exact storage-register update conditioned on the ancilla outcome.
The next recursion collects the two measurement branches with the parity sign that later appears in the estimator.
Define
\begin{equation}
\Omega_0:=\rho,
\qquad
\Omega_{\ell+1}:=\sum_{x=\pm1}x\,\mathcal{T}_x(\Omega_\ell).
\label{eq:omega_rec_supp}
\end{equation}
This parity-weighted recursion is the operator counterpart of multiplying the observed ancilla signs along a branch.
Let $Y_\ell:=\Omega_\ell\otimes \rho$.
Then
\begin{equation}
\begin{split}
\Omega_{\ell+1}
&=
\mathcal{T}_+(\Omega_\ell)-\mathcal{T}_-(\Omega_\ell)\\
&=
\Tr_F\!\left[M_+Y_\ell M_+^\dagger - M_-Y_\ell M_-^\dagger\right].
\end{split}
\end{equation}
Using \eqref{eq:swap_relations_supp}, the two Kraus terms expand as
\begin{equation}
\begin{split}
4\bigl(M_+Y_\ell M_+^\dagger - M_-Y_\ell M_-^\dagger\bigr)
&=
(W_A+W_B)Y_\ell(W_A+W_B)\\
&\quad-(W_B-W_A)Y_\ell(W_B-W_A)\\
&=
2W_AY_\ell W_B + 2W_BY_\ell W_A.
\end{split}
\end{equation}
Hence, the parity-weighted storage recursion closes exactly as
\begin{equation}
\Omega_{\ell+1}
=
\frac12
\Tr_F\!\left[
W_A(\Omega_\ell\otimes \rho)W_B
+
W_B(\Omega_\ell\otimes \rho)W_A
\right].
\label{eq:omega_closed_supp}
\end{equation}

For every history $(x_1,\dots,x_\ell)\in\{\pm1\}^\ell$, define the unnormalized branch operator
\begin{equation}
\sigma_{x_1,\dots,x_\ell}
:=
\mathcal{T}_{x_\ell}\circ\cdots\circ\mathcal{T}_{x_1}(\rho),
\end{equation}
with $\sigma_{\emptyset}:=\rho$.
By construction,
\(
\Tr(\sigma_{x_1,\dots,x_\ell})=\Prob[x_1,\dots,x_\ell].
\)
These branch operators record the full measurement tree before the parity signs are summed.
We claim that
\begin{equation}
\Omega_\ell
=
\sum_{x_1,\dots,x_\ell\in\{\pm1\}}
\left(\prod_{r=1}^{\ell}x_r\right)\sigma_{x_1,\dots,x_\ell}.
\label{eq:omega_sigma_supp}
\end{equation}
The case $\ell=0$ is immediate.
Assume that \eqref{eq:omega_sigma_supp} holds at level $\ell$. This yields
\begin{equation}
\begin{split}
\Omega_{\ell+1}
&=
\sum_{x=\pm1}x\,\mathcal{T}_x(\Omega_\ell)\\
&=
\sum_{x=\pm1}x\,\mathcal{T}_x
\left[
\sum_{x_1,\dots,x_\ell}
\left(\prod_{r=1}^{\ell}x_r\right)\sigma_{x_1,\dots,x_\ell}
\right]\\
&=
\sum_{x_1,\dots,x_{\ell+1}}
\left(\prod_{r=1}^{\ell+1}x_r\right)
\sigma_{x_1,\dots,x_{\ell+1}},
\end{split}
\end{equation}
This proves \eqref{eq:omega_sigma_supp}.
Taking the trace yields
\begin{equation}
\begin{split}
\Tr(\Omega_\ell)
&=
\sum_{x_1,\dots,x_\ell}
\left(\prod_{r=1}^{\ell}x_r\right)
\Tr(\sigma_{x_1,\dots,x_\ell})\\
&=
\sum_{x_1,\dots,x_\ell}
\left(\prod_{r=1}^{\ell}x_r\right)\Prob[x_1,\dots,x_\ell]\\
&=
\E\!\left[\prod_{r=1}^{\ell}x_r\right].
\end{split}
\label{eq:trace_omega_supp}
\end{equation}

We now introduce the operator family generated directly by the PT permutation so that it can be compared with the circuit recursion.
Define the target operator family as
\begin{equation}
\Theta_\ell
:=
\Tr_{2,\dots,\ell+1}\!\left[
\SA{\ell+1}\rho^{\otimes(\ell+1)}(\SB{\ell+1})^{-1}
\right].
\label{eq:theta_def_supp}
\end{equation}
Set
\begin{equation}
R:=\rho^{T_B}.
\end{equation}
To compare the circuit recursion with the target observable, it is useful to rewrite $\Theta_\ell$ directly in terms of powers of $R$.
The base case is
\begin{equation}
\Theta_0=\rho,
\end{equation}
since $\SA{1}=(\SB{1})^{-1}=I$.
Direct expansion gives
\begin{equation}
\begin{split}
[\Theta_\ell]_{i\mu,j\nu}
=
\sum_{\substack{a_1,\dots,a_\ell\\ \alpha_1,\dots,\alpha_\ell}}
R_{i\nu,a_1\alpha_1}
R_{a_1\alpha_1,a_2\alpha_2}
\cdots
R_{a_\ell\alpha_\ell,j\mu}
=
[R^{\ell+1}]_{i\nu,j\mu}.
\end{split}
\label{eq:theta_R_supp}
\end{equation}
Hence
\begin{equation}
\Theta_\ell=(R^{\ell+1})^{T_B}.
\label{eq:theta_power_supp}
\end{equation}

The next step is to define
\begin{equation}
\mathcal{A}_\rho(X):=\Tr_F\!\left[W_A(X\otimes \rho)W_B\right],
\qquad
\mathcal{B}_\rho(X):=\Tr_F\!\left[W_B(X\otimes \rho)W_A\right].
\label{eq:AB_maps_supp}
\end{equation}
These are exactly the two cross terms appearing in the closed recursion for $\Omega_{\ell+1}$.
Their matrix elements are
\begin{equation}
[\mathcal{A}_\rho(X)]_{i\mu,j\nu}
=
\sum_{a,\alpha} X_{a\mu,j\alpha}\,\rho_{i\alpha,a\nu},
\label{eq:A_matrix_supp}
\end{equation}
and
\begin{equation}
[\mathcal{B}_\rho(X)]_{i\mu,j\nu}
=
\sum_{a,\alpha} X_{i\alpha,a\nu}\,\rho_{a\mu,j\alpha}.
\label{eq:B_matrix_supp}
\end{equation}
Substituting $X=\Theta_\ell=(R^{\ell+1})^{T_B}$ and $\rho=R^{T_B}$ into \eqref{eq:A_matrix_supp}, we obtain
\begin{equation}
\begin{split}
[\mathcal{A}_\rho(\Theta_\ell)]_{i\mu,j\nu}
&=
\sum_{a,\alpha}
[\Theta_\ell]_{a\mu,j\alpha}\,\rho_{i\alpha,a\nu}\\
&=
\sum_{a,\alpha}
[R^{\ell+1}]_{a\alpha,j\mu}\,R_{i\nu,a\alpha}\\
&=
[R^{\ell+2}]_{i\nu,j\mu}.
\end{split}
\label{eq:A_theta_supp}
\end{equation}
Likewise, substituting into \eqref{eq:B_matrix_supp}, we get
\begin{equation}
\begin{split}
[\mathcal{B}_\rho(\Theta_\ell)]_{i\mu,j\nu}
&=
\sum_{a,\alpha}
[\Theta_\ell]_{i\alpha,a\nu}\,\rho_{a\mu,j\alpha}\\
&=
\sum_{a,\alpha}
[R^{\ell+1}]_{i\nu,a\alpha}\,R_{a\alpha,j\mu}\\
&=
[R^{\ell+2}]_{i\nu,j\mu}.
\end{split}
\label{eq:B_theta_supp}
\end{equation}
Therefore both cross terms are exactly equal and both generate the next target operator:
\begin{equation}
\mathcal{A}_\rho(\Theta_\ell)
=
\mathcal{B}_\rho(\Theta_\ell)
=
(R^{\ell+2})^{T_B}
=
\Theta_{\ell+1}.
\label{eq:theta_closed_supp}
\end{equation}

At this point the two recursions can be compared directly.
Equations~\eqref{eq:omega_rec_supp}, \eqref{eq:omega_closed_supp}, and \eqref{eq:theta_closed_supp} show that $\Omega_\ell$ and $\Theta_\ell$ obey the same recursion with the same initial condition.
Hence
\begin{equation}
\Omega_\ell=\Theta_\ell,
\qquad\text{for all }\ell\ge 0.
\end{equation}
The last step is to take the trace, which turns the operator identity into the parity expectation appearing in the estimator.
Setting $\ell=j-1$ and using \eqref{eq:trace_omega_supp}, we obtain
\begin{equation}
\E[\widehat{v}_j]
=
\Tr(\Theta_{j-1})
\;=\;
\Tr\!\bigl[(\rho_{AB}^{T_B})^j\bigr]
\;=\;
\Tr\!\bigl[\ptperm{j}\rho^{\otimes j}\bigr].
\end{equation}
The last equality is exactly the PT-permutation identity from \ref{app:pt_identity}.
\end{proof}

\section{Proof of the universal minimax converse}
\label{app:universal}

\begin{proof}
Define $\mu(\rho):=\Tr((\rho^{T_B})^K)$.
We reduce estimation to binary hypothesis testing on a commuting separable family.
The proof has two steps: first produce a PT-moment gap of size at least $4\epsmom$, and then show that distinguishing the two resulting states requires $\Omega(K/\epsmom^2)$ copies.

Work in a $2\times 2$ product subspace with computational basis
\(
\{\ket0_A,\ket1_A\}\otimes\{\ket0_B,\ket1_B\}.
\)
For $p\in(0,1)$, we define
\begin{equation}
\rho_p := p\,\ket{00}\!\bra{00} + (1-p)\,\ket{01}\!\bra{01}.
\end{equation}
Because $\rho_p$ is diagonal in a product basis, partial transpose leaves it invariant:
\(
\rho_p^{T_B}=\rho_p.
\)
Therefore, we have
\begin{equation}
\mu(\rho_p)=\Tr(\rho_p^K)=p^K+(1-p)^K.
\end{equation}
Let $\mu(p):=p^K+(1-p)^K$.

Choose
\begin{equation}
p_0:=1-\frac{1}{2K},
\qquad
p_1:=p_0+\delta,
\end{equation}
with $\delta>0$ to be specified.
Since
\begin{equation}
\mu'(p)=K\bigl(p^{K-1}-(1-p)^{K-1}\bigr),
\end{equation}
there exists an absolute constant $a>0$ such that
\begin{equation}
\mu'(p)\ge aK
\qquad
\text{for all }
p\in\left[1-\frac{1}{2K},\,1-\frac{1}{4K}\right].
\end{equation}
We choose
\begin{equation}
\delta := \frac{4\epsmom}{aK},
\end{equation}
which lies in the admissible interval for $\epsmom$ sufficiently small.
This choice makes the induced PT-moment gap strictly larger than the target estimation scale.
By the mean-value theorem, we have
\begin{equation}
|\mu(\rho_{p_1})-\mu(\rho_{p_0})|
\ge 4\epsmom.
\end{equation}
Therefore, thresholding $\widehat p_K$ at the midpoint of the two target values yields a binary test: because the two target values are separated by at least $4\epsmom$, any estimator that is $\epsmom$-accurate on both states identifies the correct state, and the resulting testing error is at most $1/3$.
The remaining task is to lower-bound the copy complexity of that testing problem.

Let $\sigma_0:=\rho_{p_0}^{\otimes \Ntotal}$ and $\sigma_1:=\rho_{p_1}^{\otimes \Ntotal}$.
By the Holevo--Helstrom theorem~\cite{Helstrom1976,Watrous2018}, success probability at least $2/3$ implies
\begin{equation}
\frac12\|\sigma_0-\sigma_1\|_1\ge \frac13.
\label{eq:helstrom_universal_supp}
\end{equation}
Because the states commute, the trace distance coincides with the $\ell_1$ distance between the corresponding Bernoulli product distributions $P_0^{\otimes \Ntotal}$ and $P_1^{\otimes \Ntotal}$.
Pinsker's inequality in the $\ell_1$/KL form together with KL additivity yields
\begin{equation}
\|\sigma_0-\sigma_1\|_1^2
\le
2\Ntotal\,D_{\mathrm{KL}}(P_0\|P_1).
\end{equation}
Combining this with \eqref{eq:helstrom_universal_supp} gives
\begin{equation}
\Ntotal \ge \frac{2}{9}\cdot \frac{1}{D_{\mathrm{KL}}(P_0\|P_1)}.
\end{equation}
In other words, \eqref{eq:helstrom_universal_supp} gives the distinguishability level required for success probability at least $2/3$, while the KL/Pinsker bound upper-bounds the distinguishability attainable from $\Ntotal$ copies; therefore $\Ntotal$ must be large enough for that upper bound to reach at least $1/3$.
For Bernoulli laws with parameters $p_0$ and $p_1=p_0+\delta$, the standard local quadratic bound gives
\begin{equation}
D_{\mathrm{KL}}(P_0\|P_1)
\le
C\,\frac{\delta^2}{\min_{p\in[p_0,p_1]}p(1-p)}
\end{equation}
for an absolute constant $C>0$.
After decreasing the absolute constant $c$ in the proposition if necessary, $p_1\le 1-1/(4K)$, and hence $p(1-p)=\Theta(1/K)$ throughout $[p_0,p_1]$. Therefore
\begin{equation}
D_{\mathrm{KL}}(P_0\|P_1)
=
O(K\delta^2)
=
O\!\left(\frac{\epsmom^2}{K}\right).
\end{equation}
Substituting the chosen value of $\delta$ proves that
\begin{equation}
\Ntotal=\Omega\!\left(\frac{K}{\epsmom^2}\right).
\end{equation}
This converse is universal.
The hard family is commuting and separable, so the proof does not use NPT-specific structure.
\end{proof}

\section{Proof of the PT-specific converse}
\label{app:pt_specific}

\begin{proof}
For every $\theta$, $\rho_\theta$ is a pure state such that
\begin{equation}
\Tr(\rho_\theta^m)=1
\qquad\text{for all }m\ge 1.
\label{eq:ordinary_constant_supp}
\end{equation}
By contrast, the PT spectrum is
\begin{equation}
\operatorname{spec}(\rho_\theta^{T_B})
=
\left\{
\cos^2\theta,\,
\sin^2\theta,\,
\sin\theta\cos\theta,\,
-\sin\theta\cos\theta
\right\},
\label{eq:pt_spectrum_supp}
\end{equation}
which depends on $\theta$ and contains a negative eigenvalue for every $\theta>0$.
Hence every state in $I_K$ is NPT.
The proof now isolates a parameter window in which the PT moments change on the scale $\sqrt{K}$ while the underlying pure states remain hard to distinguish.

Define
\begin{equation}
f_K(\theta):=\Tr\!\bigl((\rho_\theta^{T_B})^K\bigr).
\end{equation}
Using \eqref{eq:pt_spectrum_supp}, we have
\begin{equation}
f_K(\theta)=
\begin{cases}
\cos^{2K}\theta+\sin^{2K}\theta, & K \text{ odd},\\[0.4em]
\cos^{2K}\theta+\sin^{2K}\theta+2(\sin\theta\cos\theta)^K, & K \text{ even}.
\end{cases}
\label{eq:fk_formula_supp}
\end{equation}
Now we write $f_K=q_K+r_K$, where
\begin{equation}
q_K(\theta):=\cos^{2K}\theta+\sin^{2K}\theta
\end{equation}
and
\begin{equation}
r_K(\theta):=
\begin{cases}
0, & K \text{ odd},\\
2(\sin\theta\cos\theta)^K, & K \text{ even}.
\end{cases}
\end{equation}

We first show that the derivative of $f_K$ is uniformly of order $\sqrt{K}$ on $I_K$.
Since
\begin{equation}
q_K'(\theta)=2K\sin\theta\cos\theta\left(\sin^{2K-2}\theta-\cos^{2K-2}\theta\right),
\end{equation}
we have
\begin{equation}
|q_K'(\theta)|
\ge
2K\sin\theta\cos\theta\left(\cos^{2K-2}\theta-\sin^{2K-2}\theta\right).
\label{eq:qprime_lower_supp}
\end{equation}
For every $\theta\in I_K$ and $K\ge 3$, we have
\begin{equation}
\sin\theta \ge \frac{1}{6\sqrt{K}},
\qquad
\cos\theta \ge 1-\frac{\theta^2}{2}\ge \frac{31}{32},
\label{eq:sin_cos_bounds_supp}
\end{equation}
and also
\begin{equation}
\cos^{2K-2}\theta
=(\cos^2\theta)^{K-1}
\ge
(1-\theta^2)^{K-1}
\ge
1-(K-1)\theta^2
\ge
\frac{15}{16},
\label{eq:cospow_bound_supp}
\end{equation}
while
\begin{equation}
\sin^{2K-2}\theta
\le
\left(\frac{1}{4\sqrt{K}}\right)^{2K-2}
\le
\frac{1}{48}.
\label{eq:sinpow_bound_supp}
\end{equation}
Substituting \eqref{eq:sin_cos_bounds_supp}--\eqref{eq:sinpow_bound_supp} into \eqref{eq:qprime_lower_supp} gives
\begin{equation}
|q_K'(\theta)|
\ge
\frac{341}{1152}\sqrt{K}
>
\frac{1}{4}\sqrt{K}.
\label{eq:qprime_final_supp}
\end{equation}

If $K$ is even, then it holds that
\begin{equation}
r_K'(\theta)
=
2K(\sin\theta\cos\theta)^{K-1}(\cos^2\theta-\sin^2\theta),
\end{equation}
so that
\begin{equation}
|r_K'(\theta)|
\le
2K\theta^{K-1}
\le
2K\left(\frac{1}{4\sqrt{K}}\right)^{K-1}
\le
\frac{1}{12}\sqrt{K}.
\label{eq:rprime_bound_supp}
\end{equation}
When $K$ is odd, this becomes $r_K'\equiv 0$.
Combining \eqref{eq:qprime_final_supp} and \eqref{eq:rprime_bound_supp}, we conclude that for all $K\ge 3$ and all $\theta\in I_K$,
\begin{equation}
|f_K'(\theta)|\ge \frac{1}{8}\sqrt{K}.
\label{eq:derivative_lb_supp}
\end{equation}
Moreover, $q_K'(\theta)<0$ throughout $I_K$ because $\sin\theta<\cos\theta$ there.
For even $K$, the correction $r_K'(\theta)$ has the opposite sign and is strictly smaller in magnitude by \eqref{eq:rprime_bound_supp}.
Hence $f_K'$ does not change sign on $I_K$.
This proves the derivative inequality quoted in the Article with the universal choice $c_*=1/8$.

Now, we choose
\begin{equation}
\theta_0:=\frac{1}{5\sqrt{K}},
\qquad
\theta_1:=\theta_0+\Delta,
\qquad
\Delta:=\frac{32\epsmom}{\sqrt{K}}.
\end{equation}
For $\epsmom$ smaller than an absolute constant, for example $\epsmom\le 1/640$, $\theta_1$ remains inside $I_K$.
By the mean-value theorem and \eqref{eq:derivative_lb_supp}, this yields
\begin{equation}
|f_K(\theta_1)-f_K(\theta_0)|
\ge
\frac{1}{8}\sqrt{K}\,\Delta
=
4\epsmom.
\label{eq:moment_gap_supp}
\end{equation}
Therefore, thresholding $\widehat p_K$ at the midpoint of the two target values yields a strict binary discrimination reduction: because the PT-moment values differ by at least $4\epsmom$, any estimator that is $\epsmom$-accurate on both states identifies the correct hypothesis, so the resulting success probability is at least $2/3$.
The last part of the proof converts this PT-moment gap into a lower bound on the number of copies needed for discrimination.
Because the states are pure, we have
\begin{equation}
\left|\braket{\psi_{\theta_0}}{\psi_{\theta_1}}\right|
=
\cos(\theta_1-\theta_0)
=
\cos\Delta.
\end{equation}
Therefore
\begin{equation}
\frac12\left\|
\rho_{\theta_0}^{\otimes \Ntotal}
-
\rho_{\theta_1}^{\otimes \Ntotal}
\right\|_1
=
\sqrt{1-\cos^{2\Ntotal}\Delta},
\label{eq:pure_trace_distance_supp}
\end{equation}
see, e.g.,~\cite{Watrous2018}.
Now, use $1-a^N\le N(1-a)$ for $a\in[0,1]$ and $1-\cos^2\Delta=\sin^2\Delta\le \Delta^2$,
\begin{equation}
\frac12\left\|
\rho_{\theta_0}^{\otimes \Ntotal}
-
\rho_{\theta_1}^{\otimes \Ntotal}
\right\|_1
\le
\sqrt{\Ntotal}\,|\Delta|.
\label{eq:distance_bound_supp}
\end{equation}
By the Helstrom bound~\cite{Helstrom1976}, achieving success probability at least $2/3$ requires
\begin{equation}
\frac12\left\|
\rho_{\theta_0}^{\otimes \Ntotal}
-
\rho_{\theta_1}^{\otimes \Ntotal}
\right\|_1
\ge \frac13.
\end{equation}
Combining this with \eqref{eq:distance_bound_supp} yields
\begin{equation}
\sqrt{\Ntotal}\,|\Delta| \ge \frac13,
\qquad\text{hence}\qquad
\Ntotal\ge \frac{1}{9\Delta^2}.
\end{equation}
Since $\Delta=32\epsmom/\sqrt{K}$, we obtain
\begin{equation}
\Ntotal
\ge
\frac{K}{9216\,\epsmom^2}
=
\Omega\!\left(\frac{K}{\epsmom^2}\right).
\end{equation}
This proves the PT-specific converse stated in the Article.
\end{proof}

\section{Downstream partial transpose functional details}
\label{app:downstream}

This supplementary note records downstream PT-functional material supporting the separation between PT-moment acquisition and PT-functional reconstruction.

\subsection{Direct monomial \texorpdfstring{$\erfop$}{erf}--Taylor route}

Let
\begin{equation}
M_0:=\rho^{T_B},
\qquad
\Lambda\ge \|M_0\|_\infty,
\qquad
M:=M_0/\Lambda.
\end{equation}
Then $M$ is Hermitian, $\|M\|_\infty\le 1$, and
\begin{equation}
\mathcal{N}(\rho)=\frac{\|M_0\|_1-1}{2}
=
\frac{\Lambda\Tr|M|-1}{2}.
\end{equation}
The raw PT moments are
\begin{equation}
p_j:=\Tr(M_0^j),
\end{equation}
while the rescaled moments are
\begin{equation}
q_j:=\Tr(M^j)=\Lambda^{-j}p_j.
\end{equation}
We approximate $|x|$ by the smooth function
\begin{equation}
f_\alpha(x):=x\,\erfop(\alpha x),
\end{equation}
and then approximate $f_\alpha$ by a truncated Taylor polynomial in the monomial basis.

For an integer $m\ge 0$, set $D:=2m+2$ and define the degree-$D$ polynomial
\begin{equation}
p_{m,\alpha}(x)
:=
\frac{2}{\sqrt{\pi}}
\sum_{n=0}^{m}
\frac{(-1)^n\alpha^{2n+1}}{n!(2n+1)}\,x^{2n+2}
=
\sum_{j=2}^{D} c_j x^j.
\end{equation}

\begin{lemma}[$\erfop$ smoothing]
\label{lem:erf_supp}
Let $M$ have rank $r$ and spectrum inside $[-1,1]$.
For every $\delta\in(0,1]$ and every $\alpha>0$, it holds
\begin{equation}
\bigl|\Tr|M|-\Tr\!\bigl(M\,\erfop(\alpha M)\bigr)\bigr|
\le
r\delta + r\,\erfc(\alpha\delta).
\label{eq:erf_bound_supp}
\end{equation}
In particular, if $\delta=\varepsilon/(4r)$ and
\begin{equation}
\alpha
:=
\frac{4r}{\varepsilon}
\sqrt{\log\!\left(\frac{8r}{\sqrt{\pi}\varepsilon}\right)},
\end{equation}
then the right-hand side of \eqref{eq:erf_bound_supp} is at most $\varepsilon/2$.
\end{lemma}

\begin{proof}
Diagonalize $M=\sum_{i=1}^r\lambda_i\ket{\psi_i}\!\bra{\psi_i}$.
Then, it reads
\begin{equation}
\Tr|M|-\Tr\!\bigl(M\,\erfop(\alpha M)\bigr)
=
\sum_{i=1}^r |\lambda_i|\,\erfc(\alpha|\lambda_i|).
\end{equation}
Split the sum into $|\lambda_i|<\delta$ and $|\lambda_i|\ge \delta$.
The first part is at most $r\delta$.
For the second, use $|\lambda_i|\le 1$ and monotonicity of $\erfc$ to obtain at most $r\,\erfc(\alpha\delta)$.
The explicit choice of $\alpha$ follows from the standard bound
\(
\erfc(z)\le \frac{2}{\sqrt{\pi}}e^{-z^2}.
\)
\end{proof}

\begin{lemma}[Truncated Taylor remainder]
\label{lem:taylor_supp}
For every $|x|\le 1$, it holds that
\begin{equation}
\bigl|f_\alpha(x)-p_{m,\alpha}(x)\bigr|
\le
\frac{2}{\sqrt{\pi}}\alpha e^{\alpha^2}\frac{(\alpha^2)^{m+1}}{(m+1)!}.
\end{equation}
Consequently, we have
\begin{equation}
\bigl|\Tr(f_\alpha(M))-\Tr(p_{m,\alpha}(M))\bigr|
\le
\frac{2r}{\sqrt{\pi}}\alpha e^{\alpha^2}\frac{(\alpha^2)^{m+1}}{(m+1)!}.
\end{equation}
\end{lemma}

\begin{proof}
Using the Maclaurin series of $\erfop$~\cite[Ch.~7]{olver2010nist},
\begin{equation}
f_\alpha(x)-p_{m,\alpha}(x)
=
\frac{2}{\sqrt{\pi}}
\sum_{n=m+1}^{\infty}
\frac{(-1)^n\alpha^{2n+1}}{n!(2n+1)}x^{2n+2}.
\end{equation}
For $|x|\le 1$,
\begin{equation}
\bigl|f_\alpha(x)-p_{m,\alpha}(x)\bigr|
\le
\frac{2}{\sqrt{\pi}}
\alpha
\sum_{n=m+1}^{\infty}\frac{(\alpha^2)^n}{n!}
\le
\frac{2}{\sqrt{\pi}}\alpha e^{\alpha^2}\frac{(\alpha^2)^{m+1}}{(m+1)!}.
\end{equation}
Summing the scalar bound over at most $r$ eigenvalues proves the operator statement.
\end{proof}

For the monomial post-processing route, we write
\begin{equation}
\widehat{p}_j:=\frac{1}{\Mshot}\sum_{s=1}^{\Mshot}Y_j^{(s)},
\qquad
\widehat{q}_j:=\Lambda^{-j}\widehat{p}_j,
\qquad
\widehat{S}:=\sum_{j=2}^{D}c_j\widehat{q}_j,
\end{equation}
and define the rescaled coefficient weight
\begin{equation}
\lambda_\Lambda:=\sum_{j=2}^{D}|c_j|\Lambda^{-j}.
\end{equation}

\begin{proposition}[Coefficient-weight overhead for monomial \texorpdfstring{$\erfop$}{erf}--Taylor post-processing]
\label{thm:monomial_supp}
Suppose that $Y_j^{(s)}$ are simultaneous unbiased estimators of the raw PT moments $p_j=\Tr(M_0^j)$, with $|Y_j^{(s)}|\le 1$ almost surely.
Then for every $\eta>0$,
\begin{equation}
\Prob\!\left[\left|\widehat{S}-\Tr(p_{m,\alpha}(M))\right|\ge \eta\right]
\le
2\exp\!\left(-\frac{\Mshot \eta^2}{2\lambda_\Lambda^2}\right).
\label{eq:monomial_hoeffding_supp}
\end{equation}
In particular, it suffices to take
\begin{equation}
\Mshot
\ge
\frac{2\lambda_\Lambda^2}{\eta^2}
\log\frac{2}{\beta},
\label{eq:monomial_shots_supp}
\end{equation}
to guarantee probability at least $1-\beta$.
\end{proposition}

For negativity reconstruction, the relevant error split is
\begin{equation}
\bigl|\widehat{\mathcal N}-\mathcal N(\rho)\bigr|
\le
\frac{\Lambda}{2}
\left(
\bigl|\Tr|M|-\Tr(p_{m,\alpha}(M))\bigr|
+
\bigl|\Tr(p_{m,\alpha}(M))-\widehat{S}\bigr|
\right),
\label{eq:negativity_split_supp}
\end{equation}
where
\begin{equation}
\widehat{\mathcal N}:=\frac{\Lambda\widehat{S}-1}{2}.
\end{equation}
Thus negativity error $\varepsilon_N$ is ensured whenever the trace-level error for the rescaled operator is controlled to $\eta=2\varepsilon_N/\Lambda$.
The coefficient weight appearing in Proposition~\ref{thm:monomial_supp} is
\begin{equation}
\lambda_\Lambda
=
\frac{2}{\sqrt{\pi}}
\sum_{n=0}^{m}
\frac{\alpha^{2n+1}\Lambda^{-(2n+2)}}{n!(2n+1)}.
\label{eq:coeff_lambda_supp}
\end{equation}
For $\Lambda=1$, this becomes
\begin{equation}
\lambda_1
\le
\frac{2}{\sqrt{\pi}}\alpha e^{\alpha^2},
\label{eq:coeff_l1_supp}
\end{equation}
which makes the coefficient-growth source of the Hoeffding overhead explicit.
Downstream negativity complexity therefore depends on PT-moment acquisition and on the coefficient weight of the chosen polynomial representation.
This is a coefficient-stability statement for the monomial post-processing route, not an information-theoretic lower bound on all negativity estimators or polynomial representations.

\begin{proof}
The proof isolates the single-shot random variable whose average equals the polynomial functional and then applies Hoeffding's inequality.
Define
\begin{equation}
Z^{(s)}:=\sum_{j=2}^{D}c_j\Lambda^{-j}Y_j^{(s)}.
\end{equation}
Its expectation is exactly the target trace:
\begin{equation}
\E[Z^{(s)}]
=
\sum_{j=2}^{D}c_j\Lambda^{-j}p_j
=
\sum_{j=2}^{D}c_j q_j
=
\Tr(p_{m,\alpha}(M)),
\end{equation}
while the absolute value is controlled by the coefficient weight,
\begin{equation}
|Z^{(s)}|
\le
\sum_{j=2}^{D}|c_j|\Lambda^{-j}
=
\lambda_\Lambda.
\end{equation}
Applying Hoeffding's inequality to the empirical average of the $Z^{(s)}$ gives
\begin{equation}
\Prob\!\left[\left|\frac{1}{\Mshot}\sum_{s=1}^{\Mshot}Z^{(s)}-\Tr(p_{m,\alpha}(M))\right|>\eta\right]
\le
2\exp\!\left(-\frac{\Mshot\eta^2}{2\lambda_\Lambda^2}\right),
\end{equation}
which proves \eqref{eq:monomial_hoeffding_supp} and \eqref{eq:monomial_shots_supp}.
The negativity split in \eqref{eq:negativity_split_supp} follows from the definitions of the exact and estimated functionals:
\begin{equation}
\mathcal N(\rho)=\frac{\Lambda\Tr|M|-1}{2},
\qquad
\widehat{\mathcal N}=\frac{\Lambda\widehat{S}-1}{2}.
\end{equation}
Finally, the coefficient formula is obtained by inserting the monomial coefficients of $p_{m,\alpha}$:
\begin{equation}
\lambda_\Lambda
=
\frac{2}{\sqrt{\pi}}
\sum_{n=0}^{m}
\frac{\alpha^{2n+1}\Lambda^{-(2n+2)}}{n!(2n+1)}.
\end{equation}
For $\Lambda=1$, this becomes
\begin{equation}
\lambda_1
=
\frac{2}{\sqrt{\pi}}
\sum_{n=0}^{m}
\frac{\alpha^{2n+1}}{n!(2n+1)}
\le
\frac{2}{\sqrt{\pi}}\alpha e^{\alpha^2},
\end{equation}
which proves \eqref{eq:coeff_lambda_supp} and \eqref{eq:coeff_l1_supp}.
\end{proof}

\subsection{Coherent PT-LCSP route}

\begin{figure}[ht]
\centering
\includegraphics[width=0.7\linewidth]{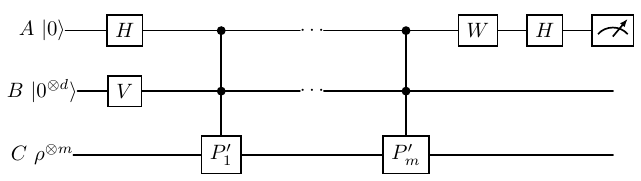}
\caption{LCSP-style coherent post-processing circuit.
It illustrates the PT-LCSP construction used in \ref{app:downstream}.}
\label{fig:lcsp_supp}
\end{figure}

Early interference-based formulations considered coherent superpositions or combinations of unitary operations to realize more general effective transformations in the duality-computing framework. In particular, different operations were applied to coherently split branches and subsequently recombined, and allowable generalized gates were expressed as linear combinations of unitary operators~\cite{long2006general,long2009allowable}.

The construction considered here uses the standard circuit-level linear-combination-of-unitaries (LCU) and Hadamard-test framework: a selector register coherently controls the PT-permutation unitaries, while their coefficients determine the selector amplitudes and phases~\cite{childs2012hamiltonian}. In the present setting, these ingredients are applied to PT-permutation moment primitives.
A schematic selector/Hadamard-test circuit for this coherent PT-LCSP route is shown in Fig.~\ref{fig:lcsp_supp}.
For the coherent route, fix a polynomial trace functional
\begin{equation}
F_D(M)=\sum_{j=2}^{D}\alpha_j \Tr\!\bigl(M^j\bigr)
=\sum_{j=2}^{D}\alpha_j\Lambda^{-j}\Tr\!\bigl(M_0^j\bigr),
\end{equation}
and define the corresponding rescaled coefficient weight
\begin{equation}
\lambda_\Lambda:=\sum_{j=2}^{D}|\alpha_j|\Lambda^{-j}.
\end{equation}
The same coherent construction can be extended to observables inserted into the trace, but we keep the statement here at the level needed for PT-moment-based polynomial post-processing.

\begin{lemma}[Permutation-trace identity]
\label{lem:permtrace_supp}
For a permutation $\pi\in S_N$ acting on $N$ copies of a state $\rho$,
\begin{equation}
\Tr(P_\pi \rho^{\otimes N})
=
\prod_{\bm c\in \mathrm{cycles}(\pi)} \Tr(\rho^{|\bm c|}),
\end{equation}
where the product is over the cycle decomposition of $\pi$.
\end{lemma}

\begin{proof}
Expand the trace in a product basis.
Each cycle of $\pi$ contracts a chain of matrix elements into one ordinary moment, and disjoint cycles factorize.
\end{proof}

\begin{lemma}[PT-LCSP implementation]
\label{lem:ptlcsp_supp}
For the functional $F_D(M)$ defined above, there exists a coherent circuit that uses $D$ copies per shot and outputs an unbiased estimator of $F_D(M)/\lambda_\Lambda$.
Estimating $F_D(M)$ to additive accuracy $\varepsilon$ requires
\begin{equation}
\Mshot = O\!\left(\frac{\lambda_\Lambda^2}{\varepsilon^2}\right)
\qquad\text{and hence}\qquad
\Ntotal = O\!\left(\frac{\lambda_\Lambda^2 D}{\varepsilon^2}\right)
\end{equation}
copies.
\end{lemma}

\begin{proof}
The selector register prepares amplitudes proportional to $\sqrt{|\alpha_j|\Lambda^{-j}/\lambda_\Lambda}$.
For real coefficients with signs, the sign of each $\alpha_j$ is absorbed into a selector-controlled phase, equivalently into the measured Hadamard-test observable.
Conditioned on selector value $j-1$ and ancilla value $\ket1$, the circuit applies $\ptperm{j}$ to the first $j$ copies.
For $j>2$, $\ptperm{j}$ need not be Hermitian, so the standard Hadamard-test $X$-measurement returns the real part of $\Tr[\ptperm{j}\rho^{\otimes j}]$. 
This is sufficient here because the PT-permutation identity identifies this quantity with $\Tr[(\rho^{T_B})^j]$, which is real.
The selector construction then realizes the polynomial as a coherent superposition of PT-permutation measurements with the correct rescaled weights.
By the PT-permutation identity from \ref{app:pt_identity}, the Hadamard-test outcome is an unbiased estimator of $F_D(M)/\lambda_\Lambda$.
Each shot is bounded in magnitude by $1$, so the same Hoeffding argument as above gives $\Mshot = O(\lambda_\Lambda^2/\varepsilon^2)$.
The coherent realization changes the implementation model and leaves the same coefficient-weight bottleneck in place.
\end{proof}

\section{Noiseless numerical validation details}
\label{app:sanity}

The noiseless checks shown in Fig.~4 of the Article use the open one-dimensional spin-chain Hamiltonian
\begin{equation}
H=\sum_{i=1}^{n-1}\left(X_iX_{i+1}+Z_iZ_{i+1}\right).
\label{eq:xz_chain_supp}
\end{equation}
We use open boundary conditions and the bipartition
\begin{equation}
A=\{1,\ldots,\lfloor n/2\rfloor\},\qquad
B=\{\lfloor n/2\rfloor+1,\ldots,n\}.
\end{equation}
For the pure-state example, the input is a ground state of Eq.~\eqref{eq:xz_chain_supp} at \(n=6\), the bipartition is \(3|3\), and moments are estimated through \(K=8\). For the mixed-state example, the input is the Gibbs state
\begin{equation}
\rho=\frac{e^{-\beta H}}{\Tr(e^{-\beta H})},\qquad \beta=1.2,
\end{equation}
at \(n=3\), the bipartition is \(1|2\), and moments are estimated through \(K=6\). The mixed state is implemented through purification in the circuit simulation.

For each data set, the number of trajectories is
\begin{equation}
N_{\mathrm{shots}}=\left\lceil30000\ln K\right\rceil,
\end{equation}
which gives \(62384\) shots for the pure-state example and \(53753\) shots for the mixed-state example. This is a heuristic choice for the numerical check and is not used in the complexity theorem. Exact moments are computed directly from the state vector or density matrix. The cumulative-parity estimates track the exact PT moments, and the residuals in Fig.~4 of the Article fluctuate around zero within the displayed statistical band. The numerical calculation is a consistency check of the analytical identity and does not contribute to the theorem statements or converse bounds.

\FloatBarrier
\section{Cloud implementation and error mitigation details}
\label{subsec:hardware}

We executed the three-qubit dynamic-circuit implementation through the IBM Quantum cloud platform~\cite{Qiskit} using the backend \texttt{ibm\_pittsburgh}, which supports mid-circuit measurement and reset.
The following details specify the remote hardware execution and mitigation procedure.

The demonstration targeted the partial-transpose moments of a generic three-qubit pure state $|\psi\rangle$ prepared by the fixed three-qubit hardware-efficient ansatz in \eqref{eq:prep_unitary}, acting on $|000\rangle$.
Concretely, the preparation unitary has the form
\begin{equation}
\begin{split}
    U_{\mathrm{prep}}
&=
\Big(\bigotimes_{i=1}^{3} R_y(\phi_i)\Big)\,
\mathrm{CX}_{1,2}\,\mathrm{CX}_{0,1}\,
\Big(\bigotimes_{i=1}^{3} R_y(\theta_i)\Big),\\
\ket{\psi}&=U_{\mathrm{prep}}\ket{000}.
\end{split}
\label{eq:prep_unitary}
\end{equation}
In the cloud hardware implementation, the ansatz architecture in \eqref{eq:prep_unitary} was kept fixed, and pseudorandom seed $42$ was used only to generate the rotation angles.
Concretely, the first layer of local rotations in \eqref{eq:prep_unitary} corresponds to
\begin{equation*}
(\phi_1,\phi_2,\phi_3)=(4.86290927,\;2.75755456,\;5.39472984),
\end{equation*}
and the second layer corresponds to
\begin{equation*}
(\theta_1,\theta_2,\theta_3)=(4.38169255,\;0.59173373,\;6.13001603).
\end{equation*}
We estimated PT moments up to fifth order ($K=5$), and computed the corresponding ideal (noise-free) baselines by classical statevector simulation of the same preparation circuit.

Because the protocol estimates all PT moments $k=2,\dots,5$ simultaneously from the same ancilla bitstrings, there is no separate shot budget for each $k$. 
Instead, the target data were obtained from two folded target circuits with noise-scaling factors $\lambda=1$ and $\lambda=3$, each executed with $30000$ shots. 
In addition, the ancilla readout calibration used two single-qubit circuits, each with $30000$ shots.
Thus the target-circuit budget was $2\times 30000=60000$ shots, together with $2\times 30000=60000$ ancilla-calibration shots.

Since the recursive circuit depth grows linearly with $K$, device noise accumulates noticeably for higher-order moments. 
To mitigate these effects, we used the following mitigation pipeline in the final data analysis. 
First, readout errors on the repeatedly measured ancilla were reduced by calibrating a single-qubit assignment matrix and applying the corresponding inverse correction to the measured count distributions. 
Second, we enabled Pauli gate twirling at the primitive level to reduce coherent error accumulation and stabilize the noise profile seen by the circuit family. 
Finally, we applied variable-noise affine regression, inspired by Clifford data regression~\cite{Czarnik2021cdr}, to map the raw PT-moment estimates to mitigated values using a training set drawn from the same ansatz family.
More specifically, the regression training set consisted of $70$ classically simulable three-qubit states generated from the same ansatz architecture with different pseudorandom angles, using seeds $1000,1001,\dots,1069$. 
For each training state, we executed two folded circuits with noise-scaling factors $\lambda=1$ and $\lambda=3$, resulting in a total of $140$ training circuits. 
Each training circuit was sampled with $6000$ shots. 
For every moment order $k=2,\dots,5$, we then fitted an affine ridge-regression model using the feature vector
\begin{equation*}
\bigl(\hat p_k^{(\lambda=1)},\hat p_k^{(\lambda=3)}\bigr),
\end{equation*}
with per-feature standardization and ridge parameter $10^{-2}$. 
No angle snapping was used in the final training set construction (i.e., the training-grid parameter was set to $0$); exact labels were obtained by classical statevector simulation at the demonstrated three-qubit scale.
The regression training stage therefore used $140\times 6000=840000$ additional shots.
At the demonstrated scale, the mitigation overhead (including regression training circuits and calibration shots) exceeds the target-circuit shot budget and should be interpreted as a proof-of-principle workflow cost rather than an asymptotic resource claim.
These implementation-specific overheads from readout calibration, Pauli twirling, and affine regression are not part of the theoretical PT-moment copy-complexity guarantee stated in the Article.

The results are shown in Fig.~5 of the Article.
The regression mitigation improves the PT-moment estimates relative to the raw data for this instance, especially at low and intermediate moment orders. 
While residual bias remains for the largest values of $k$, the mitigated data track the exact reference more closely than the unmitigated estimates, supporting compatibility of the workflow at this scale with current dynamic-circuit hardware and mitigation tooling.
This mitigation setup should be viewed as a proof-of-principle compatibility check rather than as a claim of universally transferable mitigation performance across unrelated circuit families.

\FloatBarrier
\section{Device calibration details}
\label{app:hardware_details}

This section presents the device calibration details used in the cloud hardware implementation of \ref{subsec:hardware}.
Seven qubits on \texttt{ibm\_pittsburgh} were used in the demonstration, with physical indices $(63,64,65,66,67,77,85)$, as shown in Fig.~\ref{fig:mapping}. Other single-qubit and two-qubit properties of these qubits are summarized in Table~\ref{tab:hardware_properties}. All hardware data are obtained from the IBM cloud quantum platform~\cite{Qiskit}.

\begin{figure}[H]
    \centering
    \includegraphics[width=0.62\linewidth]{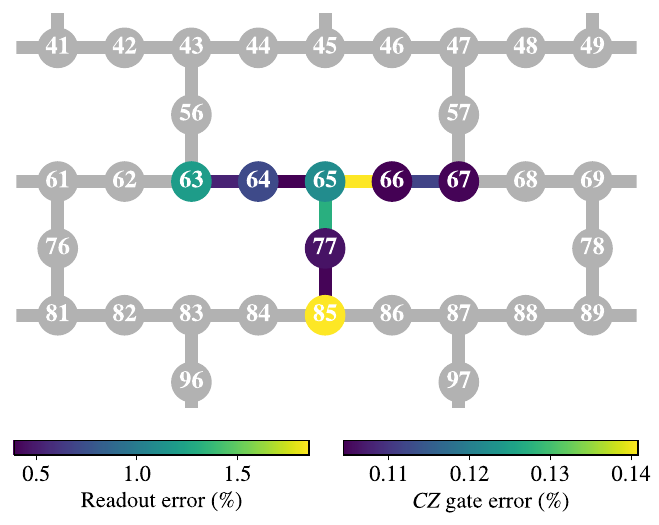}
\caption{Physical qubits on \texttt{ibm\_pittsburgh} and their coupling map. Circles represent qubits that are coloured by the readout error rates, and the edges represent the physical coupling of qubits that are coloured by the two-qubit gate error rates. Darker colour represents smaller error rates.}
    \label{fig:mapping}
\end{figure}

\begin{table}[ht]
    \caption{Device properties for the seven qubits used in the cloud hardware implementation.}
    \label{tab:hardware_properties}
    \centering
    \footnotesize
    \renewcommand{\arraystretch}{1.08}
\begin{tabular}{lcccc}
\toprule
 & Median & Mean & Min & Max \\
\midrule
$\sqrt{X}$ error (\%) & 0.012 & 0.017 & 0.010 & 0.047 \\
Readout error (\%) & 0.72 & 0.88 & 0.39 & 1.86 \\
$T_1$ ($\mu$s) & 305.17 & 304.07 & 171.75 & 378.85 \\
$T_2$ ($\mu$s) & 374.47 & 346.81 & 222.58 & 450.09 \\
CZ error (\%) & 0.11 & 0.12 & 0.10 & 0.14 \\
\bottomrule
\end{tabular}
\end{table}

The hardware circuits were generated using Qiskit's preset pass manager with optimization level $3$. 
To stabilize the regression workflow, we first transpiled the target PT-moment circuit once, extracted its final layout, and then reused this fixed layout for all subsequent target and training circuits. 
At the sampler level, Pauli gate twirling was enabled for gate operations, while measurement twirling was disabled; the twirling strategy was \texttt{active-accum}, with the number of randomizations and the shots per randomization both set to \texttt{auto}. 
For the regression model, local unitary folding was applied with scale factors $\lambda\in\{1,3\}$, and the transpiled target ISA circuit used the backend-native two-qubit gate \texttt{cz}. 
Before folding, the target ISA circuit contained $186$ two-qubit \texttt{cz} gates.

\FloatBarrier
\clearpage

\setlength{\bibsep}{0.75em}
\bibliography{references}

\end{document}